\documentclass[cleveref,thm-restate,a4paper,USenglish]{lipics-v2021}
\nolinenumbers

\usepackage{amsmath}
\usepackage{amssymb,amsfonts}
\usepackage{mathtools,amsthm}
\usepackage{thmtools,thm-restate}
\usepackage{mathrsfs}

\usepackage{enumerate}
\usepackage{xurl}
\usepackage{hyperref,xcolor}

\usepackage{multirow}
\usepackage{tabularx, environ}
\usepackage{subcaption}
\usepackage[capitalise]{cleveref}

\theoremstyle{definition}

\usepackage{mysetting}

\author{Tatsuya Gima}{Hokkaido University, Japan}{gima@ist.hokudai.ac.jp}{0000-0003-2815-5699}{JSPS KAKENHI Grant Numbers 
JP25K03077, JP26K02980, JP26K21161}

\authorrunning{T. Gima} 

\Copyright{Tatsuya Gima} 

\title{Extending Courcelle's Theorem with Optimality Predicates}
\hideLIPIcs
\titlerunning{Extending Courcelle's Theorem with Optimality Predicates}

\ccsdesc[500]{Theory of computation~Fixed parameter tractability}
\ccsdesc[500]{Theory of computation~Finite Model Theory}
\ccsdesc[500]{Theory of computation~Complexity theory and logic}
\keywords{Courcelle's theorem, Monadic second-order logic, clique-width, treewidth, optimization}

\begin{document}
\maketitle
\begin{abstract}
    Courcelle's theorem and its optimization variants yield fixed-parameter tractable algorithms for a wide range of graph problems on graphs of bounded treewidth or clique-width.
However, the limited counting power of $\mathsf{CMSO}$ poses an obstacle to capturing certain optimization problems and properties within this framework.
We introduce a new logic $\mathsf{AmCMSO}$, which extends $\mathsf{CMSO}$ with predicates for membership in the families of minimum- and maximum-cardinality sets satisfying a fixed formula $\phi(X)$.
In contrast to most previous extensions of $\mathsf{CMSO}$ with cardinality constraints,
we give algorithmic meta-theorems based on fixed-parameter tractable model checking for $\mathsf{AmCMSO}_1$ parameterized by clique-width and the formula, and for $\mathsf{AmCMSO}_2$ parameterized by treewidth and the formula.
Our proof is based on the combination of Feferman--Vaught-type decomposition and fundamental techniques for dynamic programming.
The meta-theorems yield fixed-parameter tractable algorithms for a wide range of optimization problems involving optimal solutions, including network interdiction, pre-assignment for solution uniquification, and diversity maximization, without parameterizing by the optimum value.
Finally, allowing an optimality predicate to depend on even one external set variable makes model checking hard for every level of the polynomial hierarchy, already on trees of depth four.

\end{abstract}
\clearpage

\section{Introduction}\label{sec:intro}

Courcelle's celebrated theorem~\cite{Courcelle90MSO1,CourcelleMR00} states that every fixed property expressible in counting monadic second-order logic ($\CMSO_2$) can be decided in linear time on graphs of bounded treewidth. For graphs of bounded clique-width, the corresponding result holds for $\CMSO_1$, provided that a clique-width expression of bounded width and linear size is given.
This theorem has been extended to optimization problems through the \LinE\CMSO{} framework, which solves problems of the following form in linear time on graphs of bounded clique-width, given such an expression, when $f$ is a fixed linear set function and $\phi(S)$ is a fixed \CMSO$_1$ formula.
\begin{equation*}
        \exists S  \; (f(S) \leq k \land \phi (S)).
\end{equation*}
This $\LinE\CMSO$ framework captures many combinatorial optimization problems, such as \ProblemName{Minimum Vertex Cover}, \ProblemName{Minimum Dominating Set}, and their connected variants. On graphs of bounded treewidth, its $\CMSO_2$ version also captures vertex-deletion distance to a property defined by a fixed finite set of forbidden minors~\cite{ArnborgLS91}.
Moreover, the framework extends to a set function $f$ that decomposes into element-wise sums $f(\vec S) = \sum_{v \in V} g_v(\indicator_{S_1}(v), \dots, \indicator_{S_k}(v))$ for functions $g_v \colon \{0,1\}^{k}\to \Z$~\cite{CourcelleM93}, which is not explicitly stated in the original papers but follows from a more general result.
Extensions of this framework address partitioning~\cite{Rao07}, $K$-best solutions~\cite{EppsteinK18}, diversity maximization~\cite{Baste22,DrabikM26}, the design of stable algorithms~\cite{GimaKY25}, and certain numerical constraints beyond $\MSO$~\cite{CourcelleD16}.
On graphs of bounded treewidth, these frameworks can be extended to $\CMSO_2$, which can describe properties involving edges and edge sets. In contrast, there is a fixed $\CMSO_2$ property whose model-checking problem cannot be solved in polynomial time even on cliques, which have clique-width at most 2, unless $\mathrm{E}=\mathrm{NE}$~\cite{CourcelleMR00,Lampis14}.
However, these frameworks do not directly support quantification over optimal solutions within the logic.

Several types of combinatorial optimization problems involving optimal solutions have also been studied actively, including robust optimization~\cite{Kasperski16}, network interdiction~\cite{Smith13}, uniquification problems~\cite{HoriyamaKOSS24,AnCCKLOS25}, quantified integer programming~\cite{ChistikovH17,NguyenP22}, and other multilevel optimization problems.
The complexity of most of these problems climbs one level of the polynomial hierarchy, typically from P to NP-hard or from NP to $\Sigma_2^P$-hard~\cite{GruneW25,GruneW26}.
Most of these problems that are in $\Sigma_2^P$ have a min-max structure, 
and can be expressed in the following form
\begin{equation} \tag{A}\label{eq:min-max-form} 
    \exists S \;(f(S) \le p \land \forall Y \;(g(S, Y) \le q \to \Psi(S,Y)))\;
\end{equation} with evaluation functions $f$ and $g$, and property $\Psi$. 
One example is \ProblemName{Length-Bounded Cut}, also called \ProblemName{$k$-Bounded Cut}, a network-interdiction problem.
The problem asks for a minimum-cardinality set $S$ of edges whose removal ensures that there is no $s$-$t$ path of length at most $k$ in the graph.
The form \eqref{eq:min-max-form} captures the decision version of this problem with $f(S) = |S|$, $g(S, Y) = |Y|$, and $\Psi(S,Y)$ requiring that $Y$ is not an $s$-$t$ path or $S \cap Y \neq \varnothing$.
The fixed-parameter tractability of \ProblemName{Length-Bounded Cut}, parameterized by treewidth\footnote{This application requires $\MSO_2$, so the bounded-clique-width case is nontrivial.} plus $q$, follows readily from the optimization version of Courcelle's theorem because $|Y| \le q$ can be expressed by an $\MSO$ formula of length $O(q)$; a more efficient algorithm is also known~\cite{DvorakK18}.
However, when the length bound is part of the input, \ProblemName{Length-Bounded Cut} is W[1]-hard parameterized by treewidth~\cite{DvorakK18}.

These facts imply that extending $\CMSO$ with cardinality comparisons such as $|X| \le k$ can express many hard problems even on graphs of bounded treewidth or clique-width.
For the same reason, several proposed extensions of $\MSO$ with cardinality comparisons, such as $\mathsf{cardMSO}$~\cite{GanianO13}, $\MSO$-$\mathsf{LCC}$~\cite{Szeider11}, $\MSO_{\mathsf{Lin}}^{\mathsf{GL}}$~\cite{KnopKMT19}, and ${\setlength\fboxsep{1pt}\boxed{\forall}}\MSO$~\cite{DreierGH25}, are known to be hard even on graphs of bounded treewidth or clique-width.
Consequently, these works mainly study XP algorithms~\cite{DreierGH25,KnopKMT19,Szeider11} or approximation algorithms~\cite{DreierGH25} parameterized by treewidth or clique-width, or fixed-parameter tractability~\cite{GanianO13,KnopKMT19} parameterized by more restrictive graph parameters such as vertex cover number.

Nevertheless, some problems of the form~\eqref{eq:min-max-form} are fixed-parameter tractable parameterized by clique-width alone when the bound $q$ is the optimum value.
One example is \ProblemName{Pre-assignment for Uniquification of Minimum Vertex Cover} (\ProblemName{PAU-VC})~\cite{AnCCKLOS25}, which asks for a minimum-cardinality vertex set $S$ contained in exactly one minimum vertex cover of $G$.
This problem does not appear to be captured directly by the $\LinE\CMSO$ framework, since minimum cardinality cannot in general be expressed in \MSO{} logic.
However, uniqueness can be expressed in $\CMSO$ logic, and \ProblemName{Minimum Vertex Cover} is a basic example captured by the optimization variant of Courcelle's theorem.

These results suggest a substantial complexity gap between quantifying over sets of size at most an input bound $k$ and quantifying over minimum-cardinality sets satisfying a fixed formula $\phi(Y)$.

In this paper, we introduce a new logic $\Am\CMSO$ that extends $\CMSO$ logic by allowing $\Argmin$ and $\Argmax$ operations.
A new predicate $X \in \Argmin(\phi)$ means that $X$ is a minimum-cardinality set among the sets satisfying $\phi(X)$ on the given graph $G$.
This logic captures some problems of the form~\eqref{eq:min-max-form} such that the constraints $\forall Y\; (g(S, Y) \le q \to \Psi(S,Y))$ can be replaced by $\forall Y \;(Y\in \Argmin(\psi(Y)) \to \psi'(S,Y))$ with MSO-definable properties $\psi$ and $\psi'$.
We also consider an indexed version of the $\Argmin$ operator, denoted by $\Argmin^{+k}$: it selects satisfying sets whose cardinality is the $(k+1)$-st smallest distinct feasible value, so $k=0$ selects the minimum. Under an alternative interpretation, it selects satisfying sets of cardinality $\ab<\text{minimum}> + k$. 
We write $\Am_k\CMSO$ when the indices of these operators are bounded by $k$, and $\Am\CMSO$ when all indices are zero.
Our extension captures many bilevel optimization problems, including \ProblemName{PAU-VC} and \ProblemName{${}^+k$-Bounded Cut}; the latter is the variant of \ProblemName{$k'$-Bounded Cut} in which $k'=d(s,t)+k$.
For example, \ProblemName{${}^+k$-Bounded Cut} can be defined as a following $\Am_{k}\MSO_2$ formula\footnote{We define $\Argmin^{+i}$ as a ranked version of $\Argmin$ which is not defined as the family of sets achieving $\min(\cdot) + i$, and thus the formulation is not exactly equivalent to the problem.
 However, we can modify the formula to reflect the intended meaning. See~\cref{sec:pre,sec:applications} for details.}

\begin{equation*}
    \phi(S) \equiv S \in \Argmin \begin{bmatrix}
    S \subseteq E \land \forall Y \begin{bmatrix}
    Y \in \Argmin^{\le +k}(\ab<Y \text{ is an $s$-$t$ path}>)\\
    \to S \cap Y \neq \varnothing
    \end{bmatrix}
    \end{bmatrix}
\end{equation*}
Furthermore, we can handle a weighted version of the $\Argmin$ predicate $\Argmin_f$, where $f$ is an additive function, which is defined in \cref{sec:pre}.
We write $\Am_k\CMSO_{i}(\mathcal F)$ to denote the $\Am_k\CMSO$ logic with function symbols from the set $\mathcal F$.

In contrast to previous extensions of $\MSO$ with cardinality comparisons, we show that problems of the above form are fixed-parameter tractable when the width, the formula, and $k$ are treated as parameters.
If the corresponding decompositions are given, the running time is linear in the size of the given decomposition when the parameters are considered as constants.
Combining our algorithms and the known methods for obtaining clique-width and treewidth decompositions, we obtain the following results.
\begin{restatable}{theorem}{evalargmso} \label{thm:eval-argmso}
    Let $G$ be an $n$-vertex graph with clique-width $\cw$.
    Let $\mathcal F$ be a finite set of additive functions $f \colon \pset_{\arity(f)}(V(G)) \to \Z$.
    Let $\phi$ be an $\Am_k\CMSO_1(\mathcal F)$ formula with free set variables $\vec X$.
    Then, a tuple of sets $\vec S$ satisfying $G \models \phi[\vec S]$ can be found, or its nonexistence reported, in time $g(|\phi|,k, \cw)\cdot n^2$, where $g$ is a computable function.
\end{restatable}

\begin{restatable}{theorem}{evalargmsotw} \label{thm:eval-argmso-tw}
    Let $G$ be an $n$-vertex graph with treewidth $\tw$.
    Let $\mathcal F$ be a finite set of additive functions $f \colon \pset_{\arity(f)}(V(G) \cup E(G)) \to \Z$.
    Let $\phi$ be an $\Am_k\CMSO_2(\mathcal F)$ formula with free set variables $\vec X$.
    Then, a tuple of sets $\vec S$ satisfying $G \models \phi[\vec S]$ can be found, or its nonexistence reported, in time $g(|\phi|, k, \tw)\cdot n^2$, where $g$ is a computable function.
\end{restatable}

Our framework implies that, for many problems, an FPT algorithm parameterized by clique-width plus a cardinality bound $k$ can be strengthened to one parameterized by clique-width plus $|k-\mathrm{opt}|$.
For example, \ProblemName{Diverse Minimum Vertex Cover} (DMVC) asks for $r$ minimum vertex covers $S_1,\dots, S_r$ that maximize the sum of their pairwise Hamming distances.
DMVC is known to be fixed-parameter tractable parameterized by $w+r+\tau$~\cite{Baste22,DrabikM26}, where $w$ is treewidth or clique-width and $\tau$ is the size of a minimum vertex cover.
Our results remove the dependence on $\tau$; \cref{sec:applications} discusses the details and other applications.
Furthermore, our framework handles weighted versions of these problems, which cannot be captured by the $\LinE\CMSO$ framework even when the bounds in the weighted cardinality constraints are constants.

Additionally, we consider an extension of $\Am\MSO$, called 1-$\Am\MSO$, that allows an $\Argmin$ operator to depend on one external set variable and can express formulas of the form $\phi(X) \equiv \forall Y(Y \in\Argmin(\beta(X, Y) \mid X) \to \alpha(X, Y))$.
See~\cref{sec:hardness} for the details.
We show that the model-checking problem for 1-$\Am\MSO$ is hard for every level of the polynomial hierarchy even on trees of fixed depth.

\begin{restatable}{theorem}{phhardargmso} \label{thm:hardness-argmso-tree}
	For any fixed integer $i \ge 1$,
	there exist 1-$\Am\MSO$-formulas $\phi_s$ and $\phi_p$ such that
	deciding whether $G \models \phi_s$ is $\SigmaP_i$-hard even on trees of depth 4,
	and deciding whether $G \models \phi_p$ is $\PiP_i$-hard even on trees of depth 4.
\end{restatable}

We also observe that the known algorithmic frameworks for $\CMSO$ properties yield a meta-theorem that captures a wide range of problems with cardinality constraints such as  \ProblemName{Equitable $k$-Coloring}~\cite{BodlaenderF05}, \ProblemName{Equitable $k$-Partition}~\cite{BodlaenderF05}, \ProblemName{Partial Dominating Set}~\cite{AminiFS11}, and \ProblemName{Bisection}~\cite{HanakaKS21}.
See \cref{subsec:cardcmso} for details.
\begin{restatable}{theorem}{cardinalityconstrainedcmso}\label{thm:cardinality-constrained-cmso}
    \ProblemName{Cardinality-Constrained $\Am\CMSO_1$ Problem} can be solved in $g(\cw,|\phi|) \cdot n^{2k+1}$ time for some computable function $g$, where $\cw$ is the clique-width of the input graph.
\end{restatable}

\subsection{Algorithmic Techniques}
We give an overview of the techniques behind our algorithmic results.
Our algorithm is based on the Feferman--Vaught theorem~\cite{FefermanV59,Makowsky04}, a key tool for obtaining dynamic-programming algorithms for $\MSO$ model checking on graphs of bounded clique-width or treewidth.
An extension of the Feferman--Vaught theorem to $\MSO$ states that the truth value of any $\MSO$ sentence $\phi$ on the disjoint union of two graphs $G$ and $H$ can be determined from the truth values, on $G$ and $H$, of $\MSO$ formulas whose quantifier rank is at most that of $\phi$.
Here, the quantifier rank is a complexity measure of a formula, defined as the maximum depth of nested quantifiers.
A fundamental fact of finite model theory is that, for a fixed finite signature and a fixed finite set of free variables, the number of $\MSO$ formulas of quantifier rank at most $q$ is finite up to logical equivalence.
Thus, the Feferman--Vaught theorem implies that the truth value of $\phi$ on the disjoint union of two graphs can be determined by a finite number of truth values of $\MSO$ formulas on the two graphs.

We extend the Feferman--Vaught theorem for $\MSO$ to our $\Am\MSO$ logic and show that $\Am\MSO$ model checking can be solved by dynamic programming on a decomposition tree of the input graph.
We first describe an intuitive obstacle to extending the Feferman--Vaught theorem to $\MSO$ with cardinality evaluation.
Consider a (hypothetical) dynamic programming algorithm with a table $\mathcal E\colon \Pi\times \mathbb N \to (2^{V} \to \{\top, \bot\})$, where $\Pi$ is a finite set of properties of subsets of $V$ (e.g., properties defined by $\MSO$ formulas $\phi(X)$ of bounded quantifier rank).
Set $\mathcal E_{G}[\pi, n][S] = \top$ if the set $S$ satisfies the property $\pi \in \Pi$ and $|S| = n$, and set $\mathcal E_{G}[\pi, n][S] = \bot$ otherwise.
Assume that we can update $\mathcal E$ as follows, using a function $F: \Pi \to 2^{\Pi \times \Pi}$:
\[
    \mathcal E_{G \oplus H}[\pi, n][S] = \bigvee_{(\pi_1, \pi_2) \in F(\pi)} \bigvee_{n_1 + n_2 = n} (\mathcal E_G[\pi_1, n_1][S \cap V(G)] \land \mathcal E_H[\pi_2, n_2][S \cap V(H)]).
\]
Here, the number of disjuncts is unbounded when $n$ is unbounded, since the number of pairs of nonnegative integers $n_1,n_2$ with $n_1+n_2=n$ grows with $n$.
This is an obstacle to extending the Feferman--Vaught theorem to $\MSO$ with cardinality evaluation.

Our strategy is based on the two fundamental dynamic-programming techniques: storing only the minimum feasible cardinality and reconstructing an optimal solution by backtracking through the computation.
Let $n_{\min}^{G}(\pi) = \min\{|S| : \mathcal E_G[\pi, |S|][S] = \top\}$, with $\min\varnothing=\infty$. We consider the restricted table $\mathcal E' \colon \Pi \to (2^{V} \to \{\top, \bot\})$ defined by $\mathcal E_G'[\pi][S] = \mathcal E_G[\pi, n_{\min}^{G}(\pi)][S]$ when $n_{\min}^{G}(\pi)$ is finite, and by $\mathcal E_G'[\pi][S]=\bot$ otherwise.
Let $F'_{G,H}: \Pi \to 2^{\Pi \times \Pi}$ return the pairs $(\pi_1, \pi_2)\in F(\pi)$ that minimize $n_{\min}^{G}(\pi_1) + n_{\min}^{H}(\pi_2)$ among pairs with finite values. This function depends on $G$ and $H$, and it yields the update rule
\[
    \mathcal E_{G \oplus H}'[\pi][S] = \bigvee_{(\pi_1, \pi_2) \in F'_{G,H}(\pi)} (\mathcal E_G'[\pi_1][S \cap V(G)] \land \mathcal E_H'[\pi_2][S \cap V(H)]).
\]
Hence, if $F'_{G,H}$ is known, we can determine whether a set is a minimum-cardinality set satisfying $\pi$ on $G \oplus H$ using finitely many properties on $G$ and $H$.
Using standard dynamic programming for $\MSO$ model checking, we can compute the table $F'_{G,H}$ and thus obtain a dynamic programming algorithm for $\Am\MSO$ model checking on graphs of bounded clique-width or treewidth.
Combining this procedure with the structural induction in the proof of the Feferman--Vaught theorem, we show that $\Am\MSO_1$ model checking is fixed-parameter tractable parameterized by clique-width and the formula, and that the corresponding result holds for $\Am\MSO_2$ and treewidth.

\section{Preliminaries}\label{sec:pre}
All objective functions and weight functions in this paper are integer-valued.
For simplicity, we assume that addition and comparison of two integers take constant time.
When large integers are involved, the running-time bounds in this paper can instead be interpreted as bounds on the number of arithmetic operations.

A problem is \emph{fixed-parameter tractable} parameterized by $p$ if it can be solved in $f(p)n^c$ time, where $f$ is a computable function, $n$ is the input length, and $c$ is a constant independent of $p$.
A problem is in \emph{XP} parameterized by $p$ if it can be solved in $n^{f(p)}$ time, where $f$ is a computable function and $n$ is the input length.

Let $U$ be a set.
We denote the power set of $U$ by $\pset(U)$ or $2^{U}$.
For $k\in \Z_{> 0}$, let $\pset_k(U)  = (2^{U})^k$.

Let $\vec x$ be a tuple.
The $i$-th element of $\vec x$ is denoted by $x_i$.
Let $l,r$ be integers.
The integer interval $\{i \in \mathbb Z : l \le i \le r\}$ is denoted by $[l,r]$.
We denote $[1,r]$ by $[r]$.
Let $f$ be a function from a domain $D$ to a codomain $C$.
For $S \subseteq D$, let $f(S) = \{f(s) : s \in S\}$.

We consider extended $\min$ operations, which denote the $i$-th smallest element of a set.
\begin{definition}\label{def:imin}
	Let $D = \{d_1, \ldots, d_n\} \subseteq \Z \cup \{\infty\}$ with $d_i < d_{i+1}$ for all $i \in [n-1]$.
	For $i \in [0, n-1]$, $d_{i+1}$ is the \emph{$(i+1)$-st smallest element} of $D$ and is denoted by $\imin^{+i}(D)$.
	For $i \ge n$, set $\imin^{+i}(D) = \infty$. 
\end{definition}
Here, $\imin^{+0}(D) = \min(D)$, and the superscript $+i$ indicates the distance from the minimum element on the total order of $D$.
The operator $\imax^{+i}(D)$ is defined analogously. By duality, we focus on $\imin^{+i}(D)$.

\begin{definition}\label{def:additive}
	Let $D$ be a finite set and $k\in \mathbb Z_{>0}$.
	For a set $S\subseteq D$ and $v \in D$, define the \emph{indicator function} $\indicator_S(v) = 1$ if $v \in S$ and $\indicator_S(v) = 0$ otherwise.
    A function $f\colon \pset_k(D) \to \Z$ is \emph{additive}
    if there exists a tuple of functions $\ab<g_v\colon \{0,1\}^k \to \Z>_{v \in D}$ such that $f(\vec S) = \sum_{v \in D} g_v(\indicator_{S_1}(v), \ldots, \indicator_{S_k}(v))$ for all $\vec S \in \pset_k(D)$ and $g_v(0, \dots, 0) = 0$.
\end{definition}

Weighted cardinality is one example. If $w\colon D \to \Z$ is a weight function, then
    $f(\vec S) = \sum_{i=1}^{k} \sum_{v \in S_i} w(v)$
is additive.
In this paper, we assume that the additive function $f$ is given by the tuple of functions $\ab<g_v>_{v \in D}$.

From the definition, the following property of additive functions is immediate.
\begin{observation}\label{obs:separable-additive}
	Let $D$ be a finite set and $k\in \mathbb Z_{>0}$.
	Let $f\colon \pset_k(D) \to \Z$ be an additive function on tuples of sets.
	Then, for any $D_1, D_2 \subseteq D$ with $D_1 \cap D_2 = \varnothing$, and any $\vec S \in \pset_k(D_1)$ and $\vec T \in \pset_k(D_2)$, we have $f(\ab<S_i \cup T_i>_{i \in [k]}) = f(\vec S) + f(\vec T)$.
\end{observation}

By the result of Courcelle and Mosbah~\cite{CourcelleM93} and \cref{obs:separable-additive}, if $f$ is additive, then the value $\min\{f(\vec S) : G\models \phi[\vec S]\}$ for a fixed $\CMSO_2$ formula $\phi$ can be computed in linear time on graphs of bounded treewidth.

\subsection{Logic}
A \emph{relational signature} $\sigma$ is a finite set of constant symbols and relation symbols with specified arities.
The set of first-order formulas over $\sigma$, denoted by $\FO[\sigma]$, is constructed from $R(t_1, \ldots, t_k)$, $t_1 = t_2$, $\lnot \phi$, $\phi_1 \land \phi_2$, and $\exists x\,\phi$, for each $k$-ary relation symbol $R \in \sigma$, terms $t_i$, and first-order variable $x$.

The set of monadic second-order formulas over $\sigma$, denoted by $\MSO[\sigma]$, extends $\FO[\sigma]$ with $t \in X$ and $\exists X\,\phi$ for each set variable $X$.
$\CpMSO$ extends $\MSO$ with predicates $\texttt{Card}_{m,p'}(X)$, meaning $|X| \equiv m \pmod {p'}$, for set variables $X$ and integers $m,p'$ with $2\le p'\le p$ and $0\le m<p'$.
We use $\CMSO$ to denote $\CpMSO$ for some constant $p$.
The quantifier rank of a $\CMSO$ formula $\phi$, denoted by $r_Q(\phi)$, is the maximum nesting depth of its quantifiers.
For a logic $\mathcal L$, a relational signature $\sigma$, and a nonnegative integer $r_Q$, let $\mathcal L[\sigma, r_Q]$ be the set of all $\mathcal L$ formulas over $\sigma$ whose quantifier rank is at most $r_Q$.
A \emph{sentence} is a formula without free variables.

In this paper, we consider only finite structures.
A $\sigma$-structure $\mathcal G$ consists of a domain $V$, an interpretation $R^{\mathcal G} \subseteq V^{\arity(R)}$ for each relation symbol $R \in \sigma$, and an interpretation $c^{\mathcal G}\in V$ for each constant symbol $c\in\sigma$.
Let $\vec X = \ab<X_1, \dots, X_{k_2}>$ be a tuple of $k_2$ free set variables and  $\vec x = \ab<x_1, \dots, x_{k_1}>$ ($k_1, k_2 \in \mathbb N$) be a tuple of $k_1$ free first-order variables.
Let $\phi$ be an $\mathcal L[\sigma]$ formula whose free set variables are a subset of $\vec X$ and whose free first-order variables are a subset of $\vec x$.
An \emph{assignment} of $\phi$ on $\mathcal G$ is a tuple $(\vec A, \vec a)$ of functions $\vec A \colon \vec X \to 2^{V}$ and $\vec a \colon \vec x \to V \cup \{\bot\}$.
We write $A_i = \vec A(X_i)$ for each $i \in [k_2]$ and $a_i = \vec a(x_i)$ for each $i \in [k_1]$.
An assignment $(\vec A, \vec a)$ is \emph{valid} if $a_i \neq \bot$ for every $i \in [k_1]$.
The semantics of $\mathcal G \models \phi[\vec A, \vec a]$ for a valid assignment $(\vec A, \vec a)$ is defined in the standard way for $\CMSO$ formulas.
For a sentence $\phi$, we write $\mathcal G \models \phi$ if $\mathcal G \models \phi[\varnothing]$.

To define $\Am_p\CMSO$, we extend the notion of relational signature.
A \emph{measure-enriched signature} or \emph{measured signature} $\tau = (\sigma, \mathcal F)$ is a pair consisting of a relational signature $\sigma$ and a finite set $\mathcal F$ of function symbols $f$ with arities $\arity(f)$.
Let $p$ be a nonnegative integer. The logic \emph{$\Am_p\CMSO$} over $\tau$ extends $\CMSO[\sigma]$ with the following predicates, for each $i \in [0,p]$, each $f \in \mathcal F$, each tuple of set variables $\vec X$ with $\arity(f)$ variables, and each $\Am_p\CMSO$ formula $\phi$ whose free set variables are a subset of $\vec X$ and which has no free first-order variables:
 \[ \vec X \in \Argmin^{+i}_{f}(\phi).  \]

A structure over a measured signature $\tau = (\sigma, \mathcal F)$ is a pair $\mathcal G = (\mathcal A, \mathcal F^{\mathcal A})$, where $\mathcal A$ is a $\sigma$-structure with domain $A$ and $\mathcal F^{\mathcal A} = \{f^{\mathcal A} : f \in \mathcal F\}$ is a set of functions with $f^{\mathcal A}\colon (2^{A})^{\arity(f)} \to \Z$.
The predicate $\vec X \in \Argmin^{+i}_f(\phi)$ holds when $\vec X$ satisfies $\phi$ in $\mathcal G$ and
\[
    f^{\mathcal A}(\vec X) = \imin^{+i}\ab\{f^{\mathcal A}(\vec X') : \mathcal G \models \phi[\vec X']\}.
\]
If the arity of $f$ is 1 and $f^{\mathcal A}(X) = |X|$, then we may omit the subscript $f$ and write $X \in \Argmin^{+i}(\phi)$.
Moreover, if it is clear from the context, we may omit the superscript $\mathcal A$ and write $f(X)$ instead of $f^{\mathcal A}(X)$.
We may write $\Am\CMSO$ to denote $\Am_0\CMSO$.
If $\mathcal F$ consists only of the cardinality function $X \mapsto |X|$, we identify the measured signature $(\sigma, \mathcal F)$ with the relational signature $\sigma$.
The $\Argmax$ operator is defined analogously to the $\Argmin$ operator.
Since $\Argmax$ can be expressed in terms of $\Argmin$ by negating the function values, we do not consider it as a primitive operator in our logic.

The following is a fundamental property of $\CMSO$ logic (see~\cite[Section~5.6]{CourcelleEngelfriet12}).
\begin{lemma}\label{lem:finiteness-CMSO}
	Fix a signature $\sigma$, a bound on the counting moduli, a nonnegative integer $r_Q$, and a finite set $\mathcal X$ of variables\footnote{Here, the calligraphic symbol represents a set containing both second-order and first-order variables.}.
	Let $\CMSO[\sigma, r_Q, \mathcal X]$ be the set of all $\CMSO[\sigma, r_Q]$ formulas whose free variables form a subset of $\mathcal X$.
	There is an equivalence relation $\approx$ on $\CMSO[\sigma,r_Q, \mathcal X]$ with the following properties:
	(i) If $\phi \approx \psi$, then for every $\sigma$-structure $\mathcal G$ and every assignment $(\vec A, \vec a)$, we have $\mathcal G \models \phi[\vec A, \vec a]$ if and only if $\mathcal G \models \psi[\vec A, \vec a]$;
	(ii) The number of equivalence classes of $\approx$ is finite;
	(iii) A canonical representative of the $\approx$-class of a formula is computable from that formula.
\end{lemma}

We obtain a similar property for $\Am\CMSO$.
The \emph{quantifier rank} of an $\Am\CMSO$ formula $\phi$, denoted by $r_Q(\phi)$, is the maximum nesting depth of its quantifiers.
The \emph{argmin rank} of $\phi$, denoted by $r_A(\phi)$, is the maximum nesting depth of its $\Argmin$ and $\Argmax$ operators.
We write $\Am_p\CMSO[\tau,r_Q,r_A]$ for the formulas over $\tau$ with quantifier rank at most $r_Q$ and argmin rank at most $r_A$.

\begin{lemma}
	\label{lem:finiteness-CMSOAm}
    Let $\tau = (\sigma, \mathcal F)$ be a measured signature and $p \in \pZ$.
	Fix a bound on the counting moduli, nonnegative integers $r_Q,r_A$, and a finite set $\mathcal X$ of variables.
	Let $\Am_p\CMSO[\tau,\allowbreak r_Q, r_A, \mathcal X]$ be the set of all $\Am_p\CMSO[\tau, r_Q, r_A]$ formulas whose free variables form a subset of $\mathcal X$.
	There is an equivalence relation $\approx$ on $\Am_p\CMSO[\tau,r_Q, r_A, \mathcal X]$ with the following properties:
	(i) If $\phi \approx \psi$, then for every measured $\tau$-structure $\mathcal G$ and every assignment $(\vec A, \vec a)$, we have $\mathcal G \models \phi[\vec A, \vec a]$ if and only if $\mathcal G \models \psi[\vec A, \vec a]$;
	(ii) The number of equivalence classes of $\approx$ is finite;
	(iii) A canonical representative of the $\approx$-class of a formula is computable from that formula.
\end{lemma}
\begin{proof}
	Let $\mathcal L[r_Q, r_A, \mathcal X] = \Am_p\CMSO[\tau,\allowbreak r_Q, r_A, \mathcal X]$.
	We proceed by induction on $(r_A, r_Q)$.
	Throughout the construction, we refine the equivalence classes by the exact set of free variables and choose representatives that preserve this set.
	For the base case $r_A = 0$, since $\mathcal L[r_Q, 0, \mathcal X] = \CMSO[\sigma, r_Q, \mathcal X]$, the lemma holds for any $r_Q, \mathcal X$ by \cref{lem:finiteness-CMSO}.
	Assume that the lemma holds for $\mathcal L[r_Q', r_A', \mathcal X]$ for any finite $\mathcal X$ and any $(r_A', r_Q')$ such that $r_A' < r_A$ or $(r_A' = r_A \land r_Q' < r_Q)$.
	Let $\Psi[r_Q, r_A, \mathcal X]$ be the collection of well-formed $\mathcal L[r_Q, r_A, \mathcal X]$ predicates $\psi_{i,f}(\vec Y) \equiv [\vec Y \in \Argmin^{+i}_{f}(\psi)]$, where the entries of $\vec Y$ belong to $\mathcal X$, $\psi \in \mathcal L[r_Q, r_A-1, \mathcal X]$, $i \in [0, p]$, and $f\in \mathcal F$.
	By the induction hypothesis, there are finitely many canonical representatives for $\psi$. Since there are also finitely many choices of $i$, $f$, and $\vec Y$, the collection $\Psi[r_Q, r_A, \mathcal X]$ has finitely many representatives up to logical equivalence.

	Every formula $\phi \in \mathcal L[r_Q, r_A, \mathcal X]$ can be expressed as a Boolean combination of formulas in $\Psi[r_Q, r_A, \mathcal X] \cup \mathcal L[r_Q, r_A-1, \mathcal X]$ and, when $r_Q>0$, quantified formulas $\exists X\,\psi$ and $\exists x\,\psi$ with $\psi$ in $\mathcal L[r_Q-1, r_A, \mathcal X\cup\{X\}]$ and $\mathcal L[r_Q-1, r_A, \mathcal X\cup\{x\}]$, respectively.
	After renaming the bound variables to fixed fresh symbols $X$ and $x$, the induction hypothesis gives finitely many canonical representatives for each of these collections.
	There are only finitely many Boolean functions of these representatives, so $\mathcal L[r_Q, r_A, \mathcal X]$ also has finitely many representatives up to logical equivalence.

	A canonical form can be constructed by following this induction and using a fixed ordering of the representatives and full disjunctive normal form for Boolean combinations.
\end{proof}

The \emph{length} of a formula $\phi$, denoted by $| \phi |$, is the number of symbols it contains.
Note that $r_A(\phi), r_Q(\phi) \leq | \phi |$, and the number of free variables of $\phi$ is also at most $| \phi |$.
For simplicity, we assume that $|\phi| \ge p$ where $p$ is the smallest index such that $\phi \in \Am_p\CMSO$.

\subsection{Graphs}
In this paper, by graph we mean a finite simple graph.
For a graph $G$, denote its vertex set by $V(G)$ and its edge set by $E(G)$.
Let $w$ be a positive integer.
A \emph{$w$-colored graph} $\mathscr G$ is a tuple $(G, V_1, \dots, V_w)$ where $G$ is a graph and $V_i \subseteq V(G)$ is a color class.
If there is a subset $I \subseteq [w]$ such that $V_i = \emptyset$ for each $i \in I$, then we may call $\mathscr G$ a $(w-|I|)$-colored graph.
Let $i,j\in [w]$ with $i\ne j$.
We define the graph operations $\oplus$, $\rho_{i\to j}$, and $\eta_{i,j}$ on $w$-colored graphs as follows.
The binary operation $\mathscr G_1 \oplus \mathscr G_2$ denotes the disjoint union of two $w$-colored graphs $\mathscr G_1$ and $\mathscr G_2$.
The unary operation $\rho_{i\to j} (\mathscr G)$ denotes the $w$-colored graph obtained from $\mathscr G$ by replacing $V_j$ with $V_i\cup V_j$ and $V_i$ with $\varnothing$.
The unary operation $\eta_{i,j}(\mathscr G)$ denotes the $w$-colored graph obtained from $\mathscr G$ by adding every edge between color classes $i$ and $j$.
The constant $\mathtt{i}$ denotes a $w$-colored graph with a single vertex colored with color $i$.
The \emph{clique-width} of a graph $G$ is the minimum number $w$ such that the $1$-colored graph $(G, V(G))$ can be obtained by the above operations with $w$ colors.
We write \emph{$w$-expression tree of $G$} to denote an algebraic expression tree that constructs $G$ from $w$-colored graphs.
Note that a $w$-expression tree can be viewed as a rooted tree whose internal nodes have at most two children and are labeled by $\rho_{i\to j}$, $\eta_{i,j}$, or $\oplus$.
For each fixed bound on clique-width, an expression tree of bounded width can be found in quadratic time on graphs satisfying that bound~\cite{FominK22}.

\subsection{Graph signatures}
For $\CMSO$ on graphs, the following two types of signatures are standard. See~\cite{CourcelleEngelfriet12} for more details.
A type-1 signature $\tau^1$ has one binary symmetric relation symbol $E$ and finitely many unary relation symbols and vertex constants.
For a graph $G$, the corresponding $\tau^1$-structure has domain $V(G)$ and interprets $E$ as adjacency. The unary relation symbols represent vertex colors, which need not form a partition.
Denote by $\tau^1_{w}$ a type-1 signature with $w$ unary relation symbols.
Here, a $\tau^1_{w}$ structure is identical to a $w$-colored graph, so the notions related to clique-width also apply to $\tau^1_{w}$ structures.
The logic $\CMSO[\tau^1]$ is called $\CMSO_1$ logic.

A type-2 signature $\tau^2$ has one binary relation symbol $R$, two unary relation symbols $P_E,P_V$, and finitely many additional unary relation symbols and vertex constants.
The corresponding $\tau^2$-structure of $G$ has domain $V(G) \cup E(G)$, interprets $R$ as vertex--edge incidence, and interprets $P_E$ and $P_V$ as $E(G)$ and $V(G)$, respectively.
In other words, the $\tau^2$-structure represents the incidence graph of $G$ with additional unary relations to distinguish vertices and edges.
The logic $\CMSO[\tau^2]$ is called $\CMSO_2$ logic; it extends $\CMSO_1$ by allowing quantification over edges and edge sets.

We denote by $\Am_k\CMSO_i(\mathcal F)$ the $\Am_k\CMSO_i$ logic extended with the measured signature equipped with the family of functions $\mathcal F$, for $i \in \{1,2\}$.

We can extend the notion of a $\cw$-expression to $\tau^2$-structures.
However, the clique-width of the incidence graph may be unbounded even if the original graph has constant clique-width.
Moreover, unless $\mathrm{E}=\mathrm{NE}$ (which is a kind of unary variant of $\mathrm{P}=\mathrm{NP}$), \ProblemName{$\CMSO_2$-Model Checking} has no polynomial-time algorithm even on complete graphs~\cite{CourcelleMR00,Lampis14}.
Thus, we consider only type-1 signatures for graphs of bounded clique-width.

In contrast, considering tree-decompositions or $\HR$-decompositions of $\tau^2$-structures is meaningful.
Most of our discussion of graphs of bounded treewidth generalizes to $\tau^2$-structures.
Moreover, for directed graphs, mixed graphs, and directed or undirected hypergraphs, most of our algorithmic discussion about bounded treewidth generalizes by considering the treewidth of their incidence graphs.
See the monograph of Courcelle and Engelfriet~\cite{CourcelleEngelfriet12} for more details.

Finally, we extend the operations on $w$-colored graphs to structures over a measured signature.
Let $\tau = (\tau_w^1, \mathcal F)$ be a measured signature.
A $\tau$-structure $\mathcal G = (G, \mathcal F^G)$ is \emph{additive} if every function in $\mathcal F^G$ is additive.
Let $\mathcal G = (G, \mathcal F^G)$ and $\mathcal H = (H, \mathcal F^H)$ be two additive measured $\tau$-structures.
For the operations $\mu \in \{\rho_{i\to j}, \eta_{i,j} : i,j \in [w], i\neq j\}$, define
$\mu(\mathcal G) = (\mu(G), \mathcal F^G)$.
Denote by $\mathcal F^G \sqcup \mathcal F^H$ the family of functions
$f_i^{\mathcal G\oplus \mathcal H}\colon \pset_{\arity(f_i)}(V(G) \cup V(H)) \to \Z$, for $f_i\in\mathcal F$, defined by
\[
f_i^{\mathcal G \oplus \mathcal H}(\vec S)
= f_i^{\mathcal G}(\ab<S_j\cap V(G)>_{j\in[\arity(f_i)]})
+ f_i^{\mathcal H}(\ab<S_j\cap V(H)>_{j\in[\arity(f_i)]}).
\]
Note that $f_i^{\mathcal G \oplus \mathcal H}$ is additive since $f_i^{\mathcal G}$ and $f_i^{\mathcal H}$ are additive.
Then, $\mathcal G \oplus \mathcal H = (G \oplus H, \mathcal F^G \sqcup \mathcal F^H)$.
These definitions extend the notion of a $w$-expression tree to additive measured $\tau$-structures.
A $w$-expression tree of an additive measured $\tau$-structure $\mathcal G= (G, \mathcal F^G)$ can be obtained from a $w$-expression tree of $G$ by restricting each function in $\mathcal F^G$ to the vertex set of each subexpression.

\section{Algorithm}\label{sec:alg}

\subsection{Algorithm for known \texorpdfstring{$\CMSO$}{CMSO} model checking}
First, we review the algorithm of Courcelle, Makowsky, and Rotics for $\CMSO$ model checking on graphs of bounded clique-width~\cite{CourcelleMR00}.

Let $\mathcal G$ be a structure with domain $V$.
Let $\vec X = (X_1, \dots, X_{k_2})$ be a tuple of $k_2$ variables and $\vec x = (x_1, \dots, x_{k_1})$ ($k_1, k_2 \in \mathbb N$) be a tuple of $k_1$ variables.
Let $\phi$ be a $\CMSO$ formula such that the set of its free variables is a subset of $\vec X \cup \vec x$.
Denote by $\sat(\mathcal G,\phi,\vec X, \vec x)$ the set of all valid assignments $(\vec A, \vec a)$ such that $\mathcal G \models \phi[\vec A, \vec a]$.
Since each operation $\rho_{i \to j}$ and $\eta_{i,j}$ can be defined by a quantifier-free $\MSO$ formula, the following lemma follows.
\begin{lemma}[\cite{CourcelleMR00}]
		\label{lem:translation-cw}
		Let $\cw, r_Q$ be integers.
		Let $\vec X$ be a set of set-variable symbols and $\vec x$ a set of vertex-variable symbols.
		Let $\mu \in \{\rho_{i\to j}, \eta_{i,j} : i,j \in [1,\cw], i\neq j\}$ be an operation over $\cw$-graphs.
		For every $\CMSO[\tau^{1}_{\cw}, r_Q, \vec X \cup \vec x]$ formula $\phi$, there exists a $\CMSO[\tau^{1}_{\cw}, r_Q, \vec X \cup \vec x]$ formula $\phi'$ such that for every $\cw$-graph $G$,
		\[
			\sat(\mu(G), \phi, \vec X, \vec x)  = \sat(G, \phi', \vec X, \vec x).
		\]
		Moreover, $\phi'$ can be computed in time depending only on $|\phi|$.
\end{lemma}

Thus, the only remaining case is the disjoint union operation, to which we can apply the Feferman--Vaught theorem~\cite{FefermanV59}.
There are many formulations of the Feferman--Vaught type theorems; we use the variant due to Courcelle and Engelfriet~\cite{CourcelleEngelfriet12}, which they called the \emph{splitting theorem}.
Let $\vec X$ be a set of set-variable symbols, and $\vec x, \vec y$ disjoint sets of vertex-variable symbols.
Let $D$ and $U$ be disjoint sets.
Let $\mathbf A_{I}$ be the set of all $\vec A \colon \vec X \to 2^I$ for $I \in \{D, U\}$.
Let $\mathbf a_D$ be the set of all $\vec a \colon \vec x \to D$, and let $\mathbf a_U$ be the set of all $\vec b \colon \vec y \to U$.
For $\vec a \in \mathbf a_D$ and $\vec b \in \mathbf a_U$, denote by $\vec a \concat \vec b$ the concatenation of $\vec a$ and $\vec b$, that is, $\vec a \concat \vec b \colon \vec x \cup \vec y \to D \cup U$ such that $(\vec a \concat \vec b)(x) = \vec a(x)$ for $x\in \vec x$ and $(\vec a \concat \vec b)(y) = \vec b(y)$ for $y\in \vec y$.
For $\vec A \in \mathbf A_D$ and $\vec B \in \mathbf A_U$, denote by $\vec A \cup \vec B$ the function $\vec A \cup \vec B \colon \vec X \to \pset(D\cup U)$ such that $(\vec A \cup \vec B)(X) = \vec A(X) \cup \vec B(X)$ for $X\in \vec X$.
Let $\mathcal A \subseteq \mathbf A_D \times \mathbf a_D$ and $\mathcal B \subseteq \mathbf A_U \times \mathbf a_U$.
We define $\mathcal A \boxtimes \mathcal B \coloneq \{(\vec A \cup \vec B, \vec a \concat \vec b) : (\vec A, \vec a) \in \mathcal A, (\vec B, \vec b) \in \mathcal B\}$.
We write $\bigcup_{i\in [k]} S_i$ as $\bigsqcup_{i\in [k]} S_i$ if $S_i \cap S_j = \emptyset$ for every $i, j \in [k]$ with $i\neq j$.
\begin{theorem}[{\cite[Theorem 5.39]{CourcelleEngelfriet12}}]
	\label{thm:splitting-MSO}
	Let $\vec X$ be a set of set-variable symbols and $\vec x$ a set of vertex-variable symbols.
	Let $\tau$ be a signature.
	For every $q\in \N$ and every formula $\phi(\vec X, \vec x)\in \CMSO[\tau,q, \vec X \cup \vec x]$,
	there exists a list of tuples $\FV(\phi) = \ab<\alpha_i, \beta_i, \vec y_i, \vec z_i>_{i\in [m]}$ of $\CMSO[\tau,q, \vec X \cup \vec x]$-formulas $\alpha_i$ and $\beta_i$
	such that $\vec y_i$ and $\vec z_i$ are disjoint sets of vertex-variable symbols and $\vec y_i \cup \vec z_i = \vec x$ for every $i\in [m]$,
    and for every two graphs $G$ and $H$ over the same signature $\tau$,
	\[
    \sat(G\oplus H, \phi, \vec X, \vec x) = \bigsqcup_{i\in [m]} \sat(G, \alpha_i,\vec X, \vec y_i) \boxtimes \sat(H, \beta_i, \vec X, \vec z_i).
	\]
	Moreover, the list $\FV(\phi)$ is computable in time depending only on $|\phi|$ and $\tau$.
\end{theorem}

We are now ready to describe the algorithm for $\CMSO$ model checking.
Assume that we are given a $\cw$-expression tree $\mathcal T_G$ of a graph $G$ and a $\CMSO[\tau^1_{\cw}, q]$ sentence $\phi$.
Applying \cref{thm:splitting-MSO} to the sentence $\phi$ gives
\[
G\oplus H \models \phi \iff \text{there exists } i \in [m] \text{ such that } G \models \alpha_i \text{ and } H \models \beta_i.
\]
Combining this equivalence with \cref{lem:translation-cw}, we can recursively determine whether $\sat(G, \phi)= \varnothing$ from the leaves to the root of $\mathcal T_G$.
Since the set of $\CMSO[\tau^1_{\cw}, q]$ sentences is finite up to logical equivalence by~\cref{lem:finiteness-CMSO}, we can memoize the truth value of $[\sat(G_t, \phi)=\varnothing]$ for each graph $G_t$ corresponding to a node $t$ of $\mathcal T_G$ and each $\CMSO[\tau^1_{\cw}, q]$ sentence $\phi$.
Moreover, all leaf nodes of $\mathcal T_G$ correspond to graphs with a single vertex, and thus we can compute the truth value of $[\sat(G_t, \phi)=\varnothing]$ for each leaf node $t$ and each $\CMSO[\tau^1_{\cw}, q]$ sentence $\phi$ in constant time for fixed $q$ and $\cw$.
The total number of memoized truth values is at most $f(q, \cw)\cdot n$, where $n$ is the number of nodes in $\mathcal T_G$ and $f$ is the computable function obtained from~\cref{lem:finiteness-CMSO}. Thus, we can determine whether $G \models \phi$ in time $f(q, \cw)\cdot n$.

This algorithm can be generalized as follows. See also~\cref{subsec:eval-ammso}.
Let $h_A$ be a homomorphism from $\mathscr D = (\mathcal D, \sqcup, \boxtimes)$ to $\mathscr A = (A, +_A, \times_A)$.
Then, by applying $h_A$ to the equality in \cref{thm:splitting-MSO},
we obtain recursive equations on $\mathscr A$.
For example, if $\mathscr A$ is a $(\min, +)$-semiring, we obtain an algorithm for $\LinE\CMSO_1$ problems.
Courcelle and Mosbah introduced this technique in a more general setting~\cite{CourcelleM93}; it also yields solution-counting algorithms and their variants.

\subsection{Our results}\label{sec:our-results}
In this subsection, we extend the algorithm for $\CMSO$ model checking to our logic $\Am\CMSO$.
First, we show the $\Am\CMSO$ version of \cref{lem:translation-cw}.

\begin{lemma}
		\label{lem:Amtranslation-cw}
		Let $\cw, r_Q, r_A, p \in \pZ$ and $\tau = (\tau^1_{\cw}, \mathcal F)$ be a measured signature.
		Let $\vec X$ be a set of set-variable symbols and $\vec x$ a set of vertex-variable symbols.
		Let $\mu \in \{\rho_{i\to j}, \eta_{i,j} : i,j \in [1,\cw], i\neq j\}$ be an operation over $\cw$-graphs.
		For every $\Am_p\CMSO[\tau, r_Q, r_A, \vec X \cup \vec x]$ formula $\phi$, there exists an $\Am_p\CMSO[\tau, r_Q, r_A, \vec X \cup \vec x]$ formula $\phi'$ such that for every additive measured $\tau$-structure $\mathcal G$,
		\[
			\sat(\mu(\mathcal G), \phi, \vec X, \vec x)  = \sat(\mathcal G, \phi', \vec X, \vec x).
		\]
		Moreover, $\phi'$ can be computed in time depending only on $|\phi|$.
\end{lemma}
\begin{proof}
    We prove the lemma by structural induction on $\phi$.
    The Boolean and quantifier cases follow directly from the induction hypothesis.
    For $\CMSO$ atomic formulas, we use the translation in~\cref{lem:translation-cw}.

	Assume $\phi = \vec Y \in \Argmin^{+p'}_{f}(\psi)$.
	By the induction hypothesis, there exists a formula $\psi'$ such that $\sat(\mu(\mathcal G), \psi, \vec X) = \sat(\mathcal G, \psi', \vec X)$.
	Then, $\{f(\vec Y) : \vec Y \in \sat(\mu(\mathcal G),\psi, \vec X) \} = \{f(\vec Y) : \vec Y \in \sat (\mathcal G, \psi', \vec X)\}$ and thus
	\[  \imin^{+p'} \{f(\vec Y) : \mu(\mathcal G) \models \psi[\vec Y]\} = \imin^{+p'} \{f(\vec Y) : \mathcal G \models \psi'[\vec Y]\}.
	\]
    Thus, $\vec Y \in \Argmin^{+p'}_{f}(\psi')$ is the desired formula.
\end{proof}

In the Feferman--Vaught theorem, the list $\FV(\phi)$ depends only on the formula $\phi$ and the signature $\tau$.
Thus, we can use the same list $\FV(\phi)$ for any disjoint union of two relational structures.
On the other hand, for our logic, we cannot use the same list $\FV(\phi)$ for any disjoint union of two structures.
For example, consider a formula $\phi(X)$ stating that $X$ is a maximum clique.
Then, if $G = K_4$ and $H = K_3$, for any $A \subseteq V(G) \cup V(H)$, $G \oplus H \models \phi[A]$ if and only if $G \models \phi[A]$, but if $G = K_2$ and $H = K_3$, for any $A \subseteq V(G) \cup V(H)$, $G \oplus H \models \phi[A]$ if and only if $H \models \phi[A]$.

However, we can construct a similar list $\FV_{\mathcal G, \mathcal H}(\phi)$ for any two additive measured $\tau$-structures $\mathcal G$ and $\mathcal H$ if we know their $\cw$-expressions.

\begin{restatable}{theorem}{splitting-FVMSOAM}
	\label{thm:splitting-FVMSOAM}
	Let $\cw, r_Q, r_A, p \in \pZ$ and $\tau = (\tau^1_{\cw}, \mathcal F)$ be a measured signature.
	Let $\mathcal G$ and $\mathcal H$ be two additive measured $\tau$-structures.
    For every formula $\phi\in \Am_{p}\CMSO[\tau, r_Q, r_A, \mathcal X]$ with free set variables $\vec X$ and free vertex variables $\vec x$,
	there exists a list of tuples $\FV_{\mathcal G,\mathcal H}(\phi) = \ab<\alpha_i, \beta_i, \vec y_i, \vec z_i>_{i\in [m]}$
	such that
     $\alpha_i, \beta_i\in \Am_p\CMSO[\tau, r_Q, r_A, \mathcal X]$ for every $i\in [m]$,
     $\vec y_i$ and $\vec z_i$ are disjoint sets of vertex-variable symbols and $\vec y_i \cup \vec z_i = \vec x$ for every $i\in [m]$,
	\[
    \sat(\mathcal G\oplus \mathcal H, \phi, \vec X, \vec x) = \bigsqcup_{i\in [m]} \sat(\mathcal G, \alpha_i,\vec X, \vec y_i) \boxtimes \sat(\mathcal H, \beta_i, \vec X, \vec z_i).
	\]
	Moreover, given a $\cw$-expression $\mathcal T_G$ for $\mathcal G = (G, \mathcal F^{G})$,
	the lists $\FV_{\mathcal G_1, \mathcal G_2}(\phi)$ for every term $\mathcal G_1\oplus \mathcal G_2$ in $\mathcal T_G$ and every formula $\phi$ in $\Am_{p}\CMSO[\tau, r_Q, r_A, \mathcal X]$ are computable in time $g(\tau, p, r_Q, r_A, \mathcal X) \cdot |\mathcal T_G|$,
	where $g$ is a computable function.
\end{restatable}

In \cref{subsec:proof-FV-MSO-AM}, we prove \cref{thm:splitting-FVMSOAM} by combining the Feferman--Vaught theorem with standard dynamic programming.

Combining \cref{thm:splitting-FVMSOAM,lem:finiteness-CMSOAm,lem:Amtranslation-cw}, we obtain the following result.
\begin{theorem} \label{thm:model-check-AM}
	Let $\cw, r_Q, r_A, p \in \pZ$ and $\tau = (\tau^1_{\cw}, \mathcal F)$ be a measured signature.
    Let $\mathcal X$ be a set of variable symbols.
	Given an additive measured $\tau$-structure $\mathcal G = (G, \mathcal F^G)$ with a $\cw$-expression $\mathcal T_G$ and a formula $\phi \in \Am_p\CMSO[\tau, r_Q, r_A, \mathcal X]$, we can determine whether an assignment $(\vec A, \vec a)$ satisfying $\mathcal G \models \phi[\vec A, \vec a]$ exists in time $g(\tau, |\phi|) \cdot |\mathcal T_G|$, where $g$ is a computable function.
    Moreover, if such an assignment exists, we can compute one in the same time.
\end{theorem}
\begin{proof}[Proof sketch]
    The algorithm simply checks the non-emptiness of $\sat(\mathcal G, \phi)$ by dynamic programming on the $\cw$-expression tree of $G$.
    A satisfying assignment can be recovered within the same running-time bound by standard dynamic-programming backtracking.
\end{proof}
In particular, we obtain the following as a corollary.

\evalargmso*

This theorem generalizes the $\LinE\CMSO$ framework.
Let $\phi$ be a $\CMSO$ formula with free set variables $\vec X$.
In a $\LinE\CMSO$ problem, given a graph $G$ with weight functions $w_i$ and a formula $\phi$, we want to find an assignment $\vec A$ such that $G \models \phi[\vec A]$ and $f(\vec A) \coloneq \sum_i \sum_{a \in A_i} w_i(a)$ is minimum.
Recall that $\LinE\CMSO$ problems can be solved in linear time for each fixed $\phi$ on graphs of bounded clique-width~\cite{ArnborgLS91,CourcelleMR00} if the decomposition is given.
Since $f(\vec A)$ is an additive function, we can consider the formula $\phi' = \vec X \in \Argmin_f(\phi)$ and apply \cref{thm:model-check-AM} to $\phi'$.
Then, \cref{thm:model-check-AM} implies that we can find an assignment $\vec A$ such that $G \models \phi'[\vec A]$, that is, $\vec A$ is a solution of the $\LinE\CMSO$ problem.

\subsubsection{\texorpdfstring{$\MSO_2$}{MSO2} and treewidth}
We remark that our algorithm can be applied to $\CMSO_2$ formulas on graphs of bounded treewidth.
An \emph{incidence graph} of a graph $G$ is a bipartite graph $G_I$ with the vertex set $V(G) \cup E(G)$ and the edge set $\{\{v, e\} : v \in V(G), e \in E(G), v \in e\}$.
Here, an incidence graph can be seen as a graph with 2 colors, where one color is assigned to vertices corresponding to vertices of $G$ and the other color is assigned to vertices corresponding to edges of $G$.
It is folklore that any $\CMSO_2$ formula $\phi$ on a graph $G$ is equivalent to a $\CMSO_1$ formula $\phi'$ on the incidence graph $G_I$.
Thus, if the clique-width of the incidence graph $G_I$ is small, we can apply our algorithm to $\Am\CMSO_2$ formulas.

On graphs of bounded treewidth, the treewidth of the incidence graph $G_I$ is at most the treewidth of $G$~(see, e.g.,~\cite{Kreutzer11}). Moreover, the clique-width of a 2-colored graph with treewidth $k$ is at most $f(k)$~\cite{CourcelleEngelfriet12}, where $f$ is a computable function, and such an expression tree can be constructed from a tree decomposition of the graph.
Combined with the fact that the clique-width of a graph can be bounded by a function of its treewidth~\cite{CorneilR05}, we can apply our algorithm for $\Am\CMSO_2$ formulas on graphs of bounded treewidth.

\evalargmsotw*

\subsubsection{Evaluation structure} \label{subsec:eval-ammso}
Another remark is that our algorithm can be applied to the framework of Courcelle--Mosbah~\cite{CourcelleM93}.
Their framework gives a uniform method to evaluate $\CMSO$ formulas on graphs of bounded clique-width and can produce a histogram of solution cardinalities.

An \emph{evaluation structure} is a tuple $\mathcal E = (D, \oplus, \otimes, 0_D)$ of a set $D$, two associative and commutative binary operations $\oplus$ and $\otimes$ on $D$, and an element $0_D \in D$.
Denote by $\bigoplus S$ the iterated application of $\oplus$ over the elements of $S$.
Define $\bigoplus \emptyset = 0_D$.
Fix a finite set $U$ and $k \in \mathbb Z_{>0}$.
We say that two families $A, B \in \pset(\pset_k(U))$ are \emph{separated} if there are disjoint sets $U_1, U_2 \subseteq U$ such that
$A \subseteq \pset_k(U_1)$ and $B \subseteq \pset_k(U_2)$.
A function $h\colon \pset(\pset_k(U)) \to D$ is \emph{separably evaluable} if the following conditions hold.
\begin{itemize}
	\item $h(\bigsqcup_{i} \mathscr A_i) = \bigoplus_{i} h(\mathscr A_i)$ for any pairwise disjoint families $\ab<\mathscr A_i>_i$.
	\item $h(\mathscr A \boxtimes \mathscr B) = h(\mathscr A) \otimes h(\mathscr B)$  if $\mathscr A$ and $\mathscr B$ are separated.
	\item For any singleton $\{u\} \subseteq U$, $h(\mathscr A)$ is computable in constant time for any $\mathscr A \subseteq \pset_k(\{u\})$.
\end{itemize}
For example, for $k=1$, let $h_w\colon \mathscr A \mapsto \min\{ \sum_{u \in A} w(u) : A \in \mathscr A\}$ for a weight function $w\colon U \to \Z$.
Then, $h_w$ is separably evaluable with respect to the evaluation structure $(\Z \cup \{\infty\}, \min, +, \infty)$.
Note that this $h_w$ is a composition $\min \circ f_w$ of an additive function $f_w\colon A \mapsto \sum_{u \in A} w(u)$ and $\min$.
Another example of usages of separably evaluable functions are counting the number of solutions and counting the average cardinality of the solutions~\cite{CourcelleM93}.

One useful example of a separably evaluable function is a graph polynomial.
Consider the histogram of the solutions.
The \emph{histogram} of a family $\mathcal D \subseteq \pset_k(U)$ is the function $h_{\mathcal D}\colon ([0,|U|])^k \to \pZ$ defined by
\[
h_{\mathcal D}(y_1, \dots, y_k) = |\{\ab<A_1, \dots, A_k> \in \mathcal D : |A_i| = y_i \text{ for all } i \in [1,k]\}|.
\]
Here, $h_{\mathcal D}$ can be identified with the polynomial $\sum_{\ab<A_1, \dots, A_k> \in \mathcal D} \prod_{i=1}^k x_i^{|A_i|}$ over the integer semiring $\Z$.
The \emph{spectrum} of a family $\mathcal D \subseteq \pset_k(U)$ is the function $s_{\mathcal D}\colon ([0,|U|])^k \to \{0, 1\}$ defined by $s_{\mathcal D}(y_1, \dots, y_k) = \ab[h_{\mathcal D}(y_1, \dots, y_k) > 0]$.
Then, $s_{\mathcal D}$ can be identified with the polynomial $\sum_{\ab<A_1, \dots, A_k> \in \mathcal D} \prod_{i=1}^k x_i^{|A_i|}$ over the Boolean semiring $(\{0, 1\}, \vee, \wedge, 0)$.
Let $\mathcal E$ be the polynomial semiring in $k$ variables over $\Z$.
Then, $h_{\mathcal A \sqcup \mathcal B} = h_{\mathcal A} + h_{\mathcal B}$ holds if $\mathcal A$ and $\mathcal B$ are disjoint, and $h_{\mathcal A \boxtimes \mathcal B} = h_{\mathcal A} \cdot h_{\mathcal B}$ holds if $\mathcal A$ and $\mathcal B$ are separated.
Thus, the function $\mathcal D \mapsto h_{\mathcal D}$ is separably evaluable.

Combining the proof of~\cite{CourcelleM93}, \cref{lem:Amtranslation-cw}, and \cref{thm:splitting-FVMSOAM} yields the following theorem.

\begin{theorem}\label{thm:evalargmso}
	Let $\cw, r_Q, r_A, p \in \pZ$ and $\tau = (\tau^1_{\cw}, \mathcal F)$ be a measured signature.
	Let $\phi(\vec{X})\in \Am_{p}\CMSO_1[\tau,r_Q,r_A]$ be a formula with $k$ free set variables $\vec X$.
	Let $\mathcal E = (D, \oplus_D, \otimes_D, 0_D)$ be an evaluation structure.
	Let $h\colon \pset(\pset_k(V)) \to D$ be a separably evaluable function.

    Given an additive measured $\tau$-structure $\mathcal G = (G, \mathcal F^G)$ with a $\cw$-expression $\mathcal T_G$,
	we can compute $h(\sat(\mathcal G, \phi, \vec X))$ using $g(\cw, |\phi|) \cdot |\mathcal T_G|$ arithmetic operations of $\oplus_D$ and $\otimes_D$ for some computable function $g$.
\end{theorem}
In particular, the spectrum of the solutions can be computed in time $g(\cw, |\phi|) \cdot |\mathcal T_G|\cdot O(|V(G)|^{2k})$ for some computable function $g$.

\subsection[Proof of the Feferman--Vaught-type theorem]{Proof of \cref{thm:FV-MSO-AM}}

\label{subsec:proof-FV-MSO-AM}
Our algorithm requires advice specifying the ranks of sums of the $a$-th and $b$-th optimal values.
\begin{definition}[Union-ranking]\label{def:union-rank}
    Let $D$ be a finite set, $m\in \pZ$, $k\in \mathbb Z_{>0}$, and $f\colon \pset_k(D) \to \Z$ an additive function.
    Let $D_1, D_2 \subseteq D$ be disjoint sets, and $\mathcal A_1, \dots, \mathcal A_m \in \pset(\pset_k(D_1))$, $\mathcal B_1, \dots, \mathcal B_m \in \pset(\pset_k(D_2))$.
	Let $\mathcal C = \bigcup_{i \in [1,m]} \mathcal A_i \boxtimes \mathcal B_i$.
	The \emph{union-ranking} of $\ab<\mathcal A_i, \mathcal B_i>_{i \in [1,m]}$ is the partial function $\rho_f \colon [1, m] \times \Z_{\geq 0} \times \Z_{\geq 0} \to \Z_{\geq 0}$ such that $\rho_f(i, a,b) = j$ if and only if
    $\imin^{+a}(f(\mathcal A_i)) + \imin^{+b}(f(\mathcal B_i)) = \imin^{+j}(f(\mathcal C))$ and this common value is not $\infty$.

	Let $r\in \pZ$.
	The \emph{$r$-bounded union-ranking} is the function $\rho_{f}^r \colon [1, m] \times [0,r+1] \times [0, r+1] \to [0,r+1]$ such that
    $\rho_{f}^r(i, a, b) = \rho_{f}(i, a, b)$ if $\rho_f(i,a,b)$ is defined and $\rho_{f}(i, a, b) \le r$; and $\rho_{f}^r(i, a, b) = r+1$ otherwise.
	We may omit the subscript $f$ if it is clear from the context.
\end{definition}

Note that the union-ranking is unique for a given $\ab<\mathcal A_i, \mathcal B_i>_{i \in [1,m]}$ and $f$.
The following proposition can be derived from the definition of the union-ranking.

\begin{proposition}\label{prop:fundamental-union-rank}
	Use the same notation as in \cref{def:union-rank}.
	Let $i\in [1,m]$, $t, a, b \in \mathbb Z_{\geq 0}$.
	Then, the following properties hold.
    \begin{enumerate}
        \item
        Assume $\rho_f(i, a, b) = t$.
        Then, any two tuples $A \in \mathcal A_i$ and $B \in \mathcal B_i$ with $f(A) = \imin^{+a}(f(\mathcal A_i))$ and $f(B) = \imin^{+b}(f(\mathcal B_i))$ satisfy
        $f(A \cup B) = \imin^{+t}(f(\mathcal C))$.
        \item
        For any $i\in [1,m]$, $a, b, k \in \mathbb Z_{\geq 0}$, if $\rho_f(i, a, b) \le k$ then $a, b \le k$.
    \end{enumerate}
\end{proposition}
First, we show the non-disjoint variant of \cref{thm:splitting-FVMSOAM}.
\begin{theorem}
	\label{thm:FV-MSO-AM}
	Let $\cw, r_Q, r_A, p \in \pZ$ and $\tau = (\tau^1_{\cw}, \mathcal F)$ be a measured signature.
	Let $\mathcal G$ and $\mathcal H$ be two additive measured $\tau$-structures.
	For every formula $\phi\in \Am_{p}\CMSO[\tau, r_Q, r_A, \vec X \cup \vec x]$ with free set variables $\vec X$ and free vertex variables $\vec x$,
	there exists a list of tuples $\FV_{\mathcal G,\mathcal H}(\phi) = \ab<\alpha_i, \beta_i, \vec y_i, \vec z_i>_{i\in [m]}$
	such that 
    $\alpha_i, \beta_i \in \Am_{p}\CMSO[\tau, r_Q, r_A, \vec X \cup \vec x]$ for every $i\in [m]$,
    $\vec y_i$ and $\vec z_i$ are disjoint sets of vertex-variable symbols and $\vec y_i \cup \vec z_i = \vec x$ for every $i\in [m]$,
	\[
    \sat(\mathcal G\oplus \mathcal H, \phi, \vec X, \vec x) = \bigcup_{i\in [m]} \sat(\mathcal G, \alpha_i,\vec X, \vec y_i) \boxtimes \sat(\mathcal H, \beta_i, \vec X, \vec z_i).
	\]
	Moreover, given a $\cw$-expression $\mathcal T_G$ for $\mathcal G = (G, \mathcal F^{G})$,
	the lists $\FV_{\mathcal G_1, \mathcal G_2}(\phi)$ for every term $\mathcal G_1\oplus \mathcal G_2$ in $\mathcal T_G$ and every formula $\phi$ in $\Am_{p}\CMSO[\tau, r_Q, r_A, \vec X \cup \vec x]$ are computable in time $g(\tau, |\phi|) \cdot |\mathcal T_G|$, where $g$ is a computable function.
\end{theorem}

\begin{proof}
    We proceed by two-level induction, with the outer induction following the $\cw$-expression tree $\mathcal T_G$ from the leaves to the root and the inner induction on the structure of $\phi$.
    Let $t$ be a node in $\mathcal T_G$.
    Denote by $\mathcal G_t$ the $\tau$-structure corresponding to the subexpression of $\mathcal T_G$ rooted at $t$.

    The hypothesis of the outer induction is that we have computed the vector 
        \[
        \mathscr E(\mathcal G_t,\phi, f, \vec X) \colon i \mapsto \imin^{+i}\{f^{\mathcal G_t}(\vec X) : \vec X \in \sat(\mathcal G_t, \phi, \vec X)\}
        \] 
    for each $i \in [0,p]$, $f \in \mathcal F$, a tuple of set variables $\vec X$, and $\phi\in\Am_p\CMSO[\tau, r_Q, r_A]$ with no free vertex variables.
    If the free variables of $\phi$ are not a subset of $\vec X$, or if $|\vec X|$ is not equal to the arity of $f$, then the vector $\mathscr E(\mathcal G_t,\phi, f, \vec X)$ is undefined.

    If $t$ is a leaf node, then, for fixed parameters, the vector $\mathscr E(\mathcal G_t,\phi, f, \vec X)$ can be computed in constant time for each $\phi\in\Am_p\CMSO[\tau, r_Q, r_A]$ since $\mathcal G_t$ is a graph with a single vertex.
    If $t$ has a child $t'$ and corresponds to a unary operation $\mu \in \{\rho_{i\to j}, \eta_{i,j} : i,j \in [1,\cw], i\neq j\}$, then we can compute $\mathscr E(\mathcal G_t,\phi, f, \vec X)$ from $\mathscr E(\mathcal G_{t'},\phi', f, \vec X)$ for each $\phi'\in\Am_p\CMSO[\tau, r_Q, r_A]$ and $f \in \mathcal F$ by \cref{lem:Amtranslation-cw}.

    Assume $t$ is a node corresponding to a binary operation $\mathcal G_1 \oplus \mathcal G_2$.
    Denote the vertex sets of $\mathcal G_1$ and $\mathcal G_2$ by $V_1$ and $V_2$, respectively.
    We show by induction on the structure of $\phi$ that we have computed a list $\FV_{\mathcal G_1, \mathcal G_2}(\phi) = \ab*<\ab*<\alpha_i^j, \beta_i^j>_{j\in [l_i]}, \vec y_i, \vec z_i>_{i\in [m]}$ with the following properties.
    \begin{itemize}
        \item $\vec X$ and $\vec x$ are the sets of free set variables and free vertex variables of $\phi$, respectively.
        \item $\ab<\vec y_i, \vec z_i>_{i\in [m]}$ is the list of all possible partitions of $\vec x$ into two disjoint sets.
        \item The argmin rank and the quantifier rank of $\alpha_i^j$ and $\beta_i^j$ are at most those of $\phi$.
        \item The following equality holds.
        \begin{equation*}
                \sat(\mathcal G_1\oplus \mathcal G_2, \phi, \vec X, \vec x) = \bigcup_{i\in [m]} \bigcup_{j\in [l_i]} \sat(\mathcal G_1, \alpha_i^j,\vec X, \vec y_i) \boxtimes \sat(\mathcal G_2, \beta_i^j, \vec X, \vec z_i).
        \end{equation*}
    \end{itemize}
    If $\phi$ has no free vertex variables, then we regard the list $\FV_{\mathcal G_1,\mathcal G_2}(\phi)$ as a list $\langle\alpha_i, \beta_i\rangle_{i\in [m]}$.
    For simplicity, if a partition $(\vec y, \vec z)$ is absent from the list $\FV_{\mathcal G_1,\mathcal G_2}(\phi)$, then we add the tuple $(\ab<\bot, \bot>, \vec y, \vec z)$. Here, $\bot$ denotes an always-false formula. 

    Then, we compute $\FV_{\mathcal G_1,\mathcal G_2}(\phi)$ by the following induction on the structure of $\phi$.
    From the construction of $\phi$, we only need to consider the following cases:
    $\phi \equiv \neg \psi$; $\phi \equiv \psi_1 \land \psi_2$; $\phi \equiv \exists x \psi$; $\phi \equiv \exists X \psi$; and $\phi \equiv \vec X \in \Argmin_{f}^{+p'}(\psi)$.
    If $\phi$ is an atomic formula other than $\vec X \in \Argmin_{f}^{+p'}(\psi)$, then $\phi$ is a $\CMSO$ formula and is handled by~\cref{thm:splitting-MSO}.

    \proofsubparagraph*{Case $\phi\equiv \vec X \in \Argmin_{f}^{+p'}(\psi)$:}
    From the induction hypothesis, we have a list $\FV_{\mathcal G_1,\mathcal G_2}(\psi) = \ab<\alpha_i, \beta_i>_{i\in [m]}$ that satisfies the desired properties.
    Note that $\psi$ has no free vertex variables by the definition of $\Argmin$ and thus the list $\FV_{\mathcal G_1,\mathcal G_2}(\psi)$ contains no $\vec y_i$ and $\vec z_i$.

    Let $\mathcal A_i = \sat(\mathcal G_1, \alpha_i, \vec X)$ and $\mathcal B_i = \sat(\mathcal G_2, \beta_i, \vec X)$ for each $i\in [m]$.
    Let $\rho_{f}^{p'}$ be the $p'$-bounded union-ranking of $\ab<\mathcal A_i, \mathcal B_i>_{i \in [1,m]}$.
    From the definition of the union-ranking, we have
    \begin{align*}
        \sat&(\mathcal G_1\oplus \mathcal G_2, \phi, \vec X) \\
          & = \ab\Bigg\{\vec X \in \bigcup_{i\in [m]} \mathcal A_i \boxtimes \mathcal B_i : f(\vec X) = \imin^{+p'}\ab\Bigg(f \ab\Big(\bigcup_{i\in [m]} \mathcal A_i \boxtimes \mathcal B_i))\} \\
        & = \bigcup_{\substack{
                i \in [m], (a,b) \in [0, p']^2\\
                p' = \rho_{f}^{p'}(i, a, b)
        }}
            \ab\{\vec X \in \mathcal A_i \boxtimes \mathcal B_i : f(\vec X) = \imin^{+a}(f(\mathcal A_i)) + \imin^{+b}(f(\mathcal B_i))\} \\
        & = \bigcup_{\substack{
                i \in [m], (a,b) \in [0, p']^2\\
                p' = \rho_{f}^{p'}(i, a, b)
        }}
            \ab\{\vec X \in \mathcal A_i : f(\vec X) = \imin^{+a}(f(\mathcal A_i))\} \\ &\qquad {}\boxtimes \ab\{\vec X \in \mathcal B_i : f(\vec X) = \imin^{+b}(f(\mathcal B_i))\} \\
          & = \bigcup_{\substack{
                i \in [m], (a,b) \in [0, p']^2\\
                p' = \rho_{f}^{p'}(i, a, b)
        }}
            \sat(\mathcal G_1, \vec X \in \Argmin_{f}^{+a}(\alpha_i),\vec X)
            {}\boxtimes \sat(\mathcal G_2, \vec X \in \Argmin_{f}^{+b}(\beta_i), \vec X).
    \end{align*}
    Moreover, by the outer induction hypothesis, we have computed the vectors $\mathscr E(\mathcal G_1,\alpha_i, f, \vec X)$ and $\mathscr E(\mathcal G_2,\beta_i, f, \vec X)$ for each $i\in [m]$, and thus we can compute the $p'$-bounded union-ranking $\rho_{f}^{p'}$ and the vectors $\mathscr E(\mathcal G_1\oplus \mathcal G_2, \phi, f, \vec X)$ in time $O(m(p'+1)^2)$.
    Note that $m$ is bounded by a function of $\tau$, $p$, $r_Q$, and $r_A$ by \cref{lem:finiteness-CMSOAm}.

    The other cases are similar to those in the proof of \cref{thm:splitting-MSO}~\cite{CourcelleEngelfriet12}, but we provide the details for completeness.
    \proofsubparagraph*{Case $\phi\equiv \neg \psi$:}
    Let $\FV_{\mathcal G_1 , \mathcal G_2}(\psi) = \ab*<\ab*<\alpha_i^j, \beta_i^j>_{j\in [l_i]}, \vec y_i, \vec z_i>_{i\in [m]}$.
    Denote by $U^{C}$ the complement of a set $U$.
    Then, we claim that the following equality holds, with the complement taken with respect to the set of all valid assignments for $\phi$.
    \begin{align*}
        \sat&(\mathcal G_1\oplus \mathcal G_2, \neg \psi, \vec X, \vec x) \\ & =
    \ab[\bigsqcup_{i\in [m]}\bigcup_{j\in [l_i]} \sat(\mathcal G_1, \alpha^j_i,\vec X, \vec y_i) \boxtimes \sat(\mathcal G_2, \beta^j_i, \vec X, \vec z_i)]^{C}\\
            & = \bigsqcup_{i\in [m]} \bigcup_{I\subseteq [l_i]} \sat\ab(\mathcal G_1, \bigwedge_{j\in I} \lnot\alpha^j_i,\vec X, \vec y_i) \boxtimes \sat\ab(\mathcal G_2, \bigwedge_{j\in [l_i]\setminus I}\lnot\beta^j_i, \vec X, \vec z_i).
    \end{align*}
    Here, the union over $i\in[m]$ is disjoint because distinct partitions $(\vec y_i, \vec z_i)$ yield disjoint sets of assignments.
    The first equality follows from the definition of $\lnot$. We check the second equality.
    Suppose $(\vec A, \vec a)$ is an element of the set on the first line.
    For $i=1,2$, define $\vec A_i(X)=\vec A(X)\cap V_i$ for each $X\in\vec X$, and let $\vec a_i$ be the restriction of $\vec a$ to $\{x\in\vec x : \vec a(x)\in V_i\}$.
    Then, $(\vec A, \vec a) \notin \sat(\mathcal G_1\oplus \mathcal G_2, \psi, \vec X, \vec x)$.
    Let $i \in [m]$ be the index such that the domains of $\vec a_{1}$ and $\vec a_{2}$ are $\vec y_i$ and $\vec z_i$, respectively.
    Let $I = \{j\in [l_i] : (\vec A_{1}, \vec a_{1}) \notin \sat(\mathcal G_1, \alpha^j_i,\vec X, \vec y_i)\}$.
    Then, $(\vec A_{1}, \vec a_{1}) \in \sat(\mathcal G_1, \bigwedge_{j\in I} \lnot\alpha^j_i,\vec X, \vec y_i)$.
    Since $(\vec A, \vec a)$ is an element of the set on the first line, we have $(\vec A_{2}, \vec a_{2}) \notin \sat(\mathcal G_2, \beta^j_i, \vec X, \vec z_i)$ for every $j\in [l_i]$ with $j\notin I$.
    Thus, $(\vec A_{2}, \vec a_{2}) \in \sat(\mathcal G_2, \bigwedge_{j\in [l_i]\setminus I}\lnot\beta^j_i, \vec X, \vec z_i)$.
    Hence, $(\vec A, \vec a)$ is an element of the set on the second line.

    Suppose $(\vec A, \vec a)$ is an element of the set on the second line and define $(\vec A_{i}, \vec a_{i})$ as above for $i=1,2$.
    Then, there is a unique $i\in [m]$ and at least one $I\subseteq [l_i]$ such that $(\vec A_{1}, \vec a_{1}) \in \sat(\mathcal G_1, \bigwedge_{j\in I} \lnot\alpha^j_i,\vec X, \vec y_i)$ and $(\vec A_{2}, \vec a_{2}) \in \sat(\mathcal G_2, \bigwedge_{j\in [l_i]\setminus I}\lnot\beta^j_i, \vec X, \vec z_i)$.
    Then, $(\vec A, \vec a) \notin \sat(\mathcal G_1, \alpha^j_i,\vec X, \vec y_i) \boxtimes \sat(\mathcal G_2, \beta^j_i, \vec X, \vec z_i)$ for every $j\in [l_i]$.
    Therefore, $(\vec A, \vec a)$ is an element of the set on the first line.
    Hence, the second equality holds, and we obtain the desired list $\FV_{\mathcal G_1,\mathcal G_2}(\neg \psi)$.

    \proofsubparagraph*{Case $\phi \equiv \exists x \psi$:}
    If $x$ is not a free vertex variable of $\psi$, then $\FV_{\mathcal G_1,\mathcal G_2}(\exists x\,\psi) = \FV_{\mathcal G_1,\mathcal G_2}(\psi)$. Assume otherwise.
    Let $\FV_{\mathcal G_1,\mathcal G_2}(\psi) = \ab*<\ab*<\alpha_i^j,\beta_i^j>_{j\in[l_i]},\vec y_i,\vec z_i>_{i\in[m]}$.
    For each $i \in [m]$, define the list $L_i$ by
    $L_i = \ab*<\exists x\,\alpha_i^j,\beta_i^j>_{j\in[l_i]}$ if $x \in \vec y_i$, and
    $L_i = \ab*<\alpha_i^j,\exists x\,\beta_i^j>_{j\in[l_i]}$ if $x \in \vec z_i$.
    By the induction hypothesis and the semantics of $\exists x$, we have
    \begin{align*}
        \sat&(\mathcal G_1\oplus\mathcal G_2,\exists x\,\psi,\vec X,\vec x\setminus\{x\}) \\ & =
        \begin{bmatrix}
            & \bigcup_{\substack{i\in[m],\,j\in[l_i]\\x\in\vec y_i}}
            \sat(\mathcal G_1,\exists x\,\alpha_i^j,\vec X,\vec y_i\setminus\{x\})
            \boxtimes
            \sat(\mathcal G_2,\beta_i^j,\vec X,\vec z_i) \\
            \cup & \bigcup_{\substack{i\in[m],\,j\in[l_i]\\x\in\vec z_i}}
            \sat(\mathcal G_1,\alpha_i^j,\vec X,\vec y_i)
            \boxtimes
            \sat(\mathcal G_2,\exists x\,\beta_i^j,\vec X,\vec z_i\setminus\{x\})
        \end{bmatrix} \\
        &=
        \bigcup_{i\in[m]}\ \bigcup_{(\alpha,\beta)\in L_i}
        \sat(\mathcal G_1,\alpha,\vec X,\vec y_i\setminus\{x\})
        \boxtimes
        \sat(\mathcal G_2,\beta,\vec X,\vec z_i\setminus\{x\}).
    \end{align*}
    Therefore, merging equal partitions yields the desired list
    \[
        \FV_{\mathcal G_1,\mathcal G_2}(\exists x\,\psi),
    \]
    where the merged partitions have the form $(\vec y_i \setminus \{x\}, \vec z_i \setminus \{x\})$.

    \proofsubparagraph*{Case $\phi \equiv \exists X \psi$:}
    Let
    \[
        \FV_{\mathcal G_1,\mathcal G_2}(\psi)
        =\ab*<\ab*<\alpha_i^j,\beta_i^j>_{j\in[l_i]},\vec y_i,\vec z_i>_{i\in[m]}.
    \]
    By the semantics of $\exists X$ and the induction hypothesis, we have
    \begin{align*}
        \sat(\mathcal G_1\oplus\mathcal G_2,\exists X\,\psi,\vec X\setminus\{X\},\vec x)
        &= \bigsqcup_{i\in[m]}\bigcup_{j\in[l_i]}
        \begin{pmatrix}
        \sat(\mathcal G_1,\exists X\,\alpha_i^j,\vec X\setminus\{X\},\vec y_i)\\
        {}\boxtimes  \sat(\mathcal G_2,\exists X\,\beta_i^j,\vec X\setminus\{X\},\vec z_i)
        \end{pmatrix}.
    \end{align*}
    Hence, the desired list is
    \[
    \FV_{\mathcal G_1,\mathcal G_2}(\exists X\,\psi)
        =\ab*<
            \ab*<\exists X\,\alpha_i^j,\exists X\,\beta_i^j>_{j\in[l_i]},
            \vec y_i,\vec z_i
        >_{i\in[m]}.
    \]

    \proofsubparagraph*{Case $\phi \equiv \psi_1 \land \psi_2$:}
    From the induction hypothesis, we have the following lists.
    \begin{align*}
        \FV_{\mathcal G_1,\mathcal G_2}(\psi_1)
        &=\ab*<\ab*<\alpha_i^j,\beta_i^j>_{j\in[l_i]},
        \vec y_i^1,\vec z_i^1>_{i\in[m_1]}, \\
        \FV_{\mathcal G_1,\mathcal G_2}(\psi_2)
        &=\ab*<\ab*<\gamma_k^\ell,\delta_k^\ell>_{\ell\in[r_k]},
        \vec y_k^2,\vec z_k^2>_{k\in[m_2]}.
    \end{align*}
    For each partition $(\vec y,\vec z)$ of $\vec x$, let $L_{\vec y,\vec z}$ be the list
    \[
        \ab*<
            \alpha_i^j\land\gamma_k^\ell,
            \beta_i^j\land\delta_k^\ell
        >_{\substack{
            i\in[m_1],\,j\in[l_i],\,k\in[m_2],\,\ell\in[r_k]\\
            (\vec y_i^1 \cup \vec y_k^2, \vec z_i^1 \cup \vec z_k^2)=(\vec y,\vec z)
        }}.
    \]
    For fixed formulas $\alpha_i^j,\beta_i^j,\gamma_k^\ell,\delta_k^\ell$ and a fixed partition $(\vec y,\vec z)$, we have
    \begin{align*}
        &\ab(
            \sat(\mathcal G_1,\alpha_i^j,\vec X,\vec y)
            \boxtimes
            \sat(\mathcal G_2,\beta_i^j,\vec X,\vec z)
        )
        \cap
        \ab(
            \sat(\mathcal G_1,\gamma_k^\ell,\vec X,\vec y)
            \boxtimes
            \sat(\mathcal G_2,\delta_k^\ell,\vec X,\vec z)
        ) \\
        &\quad={}
        \sat(\mathcal G_1,\alpha_i^j\land\gamma_k^\ell,\vec X,\vec y)
        \boxtimes
        \sat(\mathcal G_2,\beta_i^j\land\delta_k^\ell,\vec X,\vec z).
    \end{align*}
    It follows that
    \begin{align*}
        \sat(\mathcal G_1\oplus\mathcal G_2,\psi_1\land\psi_2,\vec X,\vec x)
        &=
        \bigsqcup_{(\vec y,\vec z)}\ \bigcup_{(\theta,\xi)\in L_{\vec y,\vec z}}
        \sat(\mathcal G_1,\theta,\vec X,\vec y)
        \boxtimes
        \sat(\mathcal G_2,\xi,\vec X,\vec z).
    \end{align*}
    Thus, $\FV_{\mathcal G_1,\mathcal G_2}(\psi_1\land\psi_2)
        =\ab*<L_{\vec y,\vec z},\vec y,\vec z>_{(\vec y,\vec z):L_{\vec y,\vec z}\neq\emptyset}$
    is the desired list.

    We finally verify the outer induction hypothesis and the running time.
    Suppose that $\phi$ has no free vertex variables and write the constructed list as
    $\FV_{\mathcal G_1,\mathcal G_2}(\phi)=\ab<\alpha_i,\beta_i>_{i\in[m]}$.

    By \cref{prop:fundamental-union-rank}, for each $t\in[0,p]$ and $f\in\mathcal F$,
    \begin{align*}
        &\mathscr E(\mathcal G_1\oplus\mathcal G_2,\phi,f,\vec X)(t)\\
        &\quad=\imin^{+t}\ab*\{
            \mathscr E(\mathcal G_1,\alpha_i,f,\vec X)(a)+\mathscr E(\mathcal G_2,\beta_i,f,\vec X)(b)
            : i\in[m],\ (a,b)\in[0,p]^2
        \}.
    \end{align*}
    Consequently, the vector $\mathscr E(\mathcal G_1\oplus\mathcal G_2,\phi,f, \vec X)$ is computable from the two child vectors using $O(m\cdot (p+1)^2)$ integer additions.

    In each construction, the quantifier rank and the argmin rank of the resulting formulas are at most those of $\phi$.
    By replacing formulas with their canonical representatives and removing duplicate pairs, \cref{lem:finiteness-CMSOAm} bounds every list by a function of $\tau$, $p$, $r_Q$, and $r_A$.
    Hence, the work at each node of $\mathcal T_G$ is bounded by $g(\tau,p,r_Q,r_A)$ for some computable function $g$, and the total running time is $g(\tau,|\phi|)\cdot|\mathcal T_G|$.
    This completes both inductions and the proof.
\end{proof}

Lastly, we modify the list $\FV_{\mathcal G, \mathcal H}(\phi)$ to make the union of assignments in \cref{thm:FV-MSO-AM} disjoint and obtain \cref{thm:splitting-FVMSOAM}.
\begin{proof}[Proof of \cref{thm:splitting-FVMSOAM}]
	The strategy is essentially the same as in the proof of~\cite[Proposition 5.37]{CourcelleEngelfriet12}.
	Let $\FV_{\mathcal G, \mathcal H}(\phi) = \ab*<\ab*<\alpha_i^j, \beta_i^j>_{j\in [l_i]}, \vec y_i, \vec z_i>_{i\in [m]}$ be the list of formulas obtained from \cref{thm:FV-MSO-AM}.
	For each $i\in[m]$, consider the formulas $\gamma_I = \bigwedge_{j \in I} \alpha_i^j \land \bigwedge_{j\in[l_i]\setminus I} \lnot \alpha_i^j$ and $\delta_I = \bigwedge_{j \in I } \beta_i^j \land \bigwedge_{j\in[l_i]\setminus I} \lnot \beta_i^j$ for each $I \subseteq [l_i]$.
    Then, the list \[\FV_{\mathcal G, \mathcal H}(\phi) = \ab*<\ab<\gamma_I, \delta_J>_{I, J \subseteq [l_i]; I\cap J \neq \varnothing}, \vec y_i, \vec z_i>_{i\in [m]}\] satisfies the desired property.
\end{proof}

\section{Applications}\label{sec:applications}
In this section, we discuss applications of our algorithm.
We often use syntactic sugar such as $X \subseteq Y$, $X = Y$, $X \cap Y = \varnothing$, and $x \in X\cap Y$; these are abbreviations with their standard meanings and are definable in $\MSO$.
We also abbreviate $\bigvee_{i \in [0,\delta']} S \in \Argmin_f^{+i}(\phi)$ by $S \in \Argmin_f^{\le +\delta'}(\phi)$.

First, we sketch how to express a bound of the form $\min(\cdot) + \delta$ using our $\imin^{+\delta'}(\cdot)$ operator, where $\delta$ is a nonnegative integer.
In our setting, we can compute $\imin^{+\delta}(\cdot)$ efficiently for a fixed $\delta$.
Let $U$ be a set, $\mathcal U \subseteq \pset(U)$, $f\colon \mathcal U \to \mathbb Z$ be a set function.
It is easy to see that $\min(f(\mathcal U)) + \delta \le \imin^{+\delta}(f(\mathcal U))$ for any $\delta \ge 0$.
Let $\delta'+1$ be the smallest positive integer such that $\min(f(\mathcal U)) +\delta < \imin^{+(\delta'+1)}(f(\mathcal U))$.
Then, we can check $f(S) \le \min(f(\mathcal U)) + \delta$ by checking $f(S) \in \{\imin^{+i}(f(\mathcal U)) : i \in [0,\delta']\}$.
Thus, the property $\ab<f(S) \le \min(f(\mathcal U)) + \delta \land \phi(S)>$ can be converted into $S \in \Argmin_f^{\le +\delta'}(\phi)$ if $\mathcal U$ is defined by an $\Am\CMSO$ formula $\phi$.
Now, $\delta' \le \delta$, which implies that if we obtain an FPT algorithm with $\delta'$ as a parameter, then we also obtain an FPT algorithm with $\delta$ as a parameter instead of $\delta'$.

\subsection{Interdiction problems}
In this subsection, we consider the following types of problems.
\begin{definition}[Interdiction Problems~{\cite{GruneW25}} on Graphs]
    Let $\phi(X)$ be a graph property of vertex sets or edge sets.
    Given a graph $G$, two weight functions $w,w'$, and thresholds $t_1,t_2\in \mathbb Z$,
    \ProblemName{$\phi$-Interdiction} asks whether the following formula holds:
    \begin{align*}
        \exists X \ab[w(X) \leq t_1 \land  \forall Y \ab((\phi(Y) \land w'(Y) \leq t_2) \to X \cap Y \neq \varnothing)].
    \end{align*}
\end{definition}
Intuitively, the task of \ProblemName{$\phi$-Interdiction} is to find a small set $X$ such that any ``nearly optimal'' solution of $\phi$ must intersect $X$.
Note that the above form captures the interdiction variants of maximization problems, by setting $w'$ to negative weights for the objective function.
Interdiction problems are also known as most vital nodes problems.

If $t_2$ is the optimum value of $w'(Y)$ such that $\phi(Y)$ holds, we call \ProblemName{$\phi$-Interdiction} the \ProblemName{Optimal-$\phi$-Interdiction} problem.
These combinatorial optimization problems have been widely investigated~\cite{GruneW25,BaierEHKKPSS10,DvorakK18,Zenklusen10,FredericksonS99,BallGV89}.
Recently, it was shown that for many properties $\phi(X)$ whose associated decision problem is NP-hard, the interdiction variant is $\Sigma^{\mathrm P}_2$-complete~\cite{GruneW25}.

Here, the constraint $|Y| \leq b$ can be expressed by a formula of length $O(b)$.
Thus, when $w'(Y)=|Y|$ and $\phi$ is a fixed $\CMSO_1$ (or $\CMSO_2$) formula, this type of problem is fixed-parameter tractable parameterized by $t_2+\cw$ (or $t_2 + \tw$), where $\cw$ is clique-width and $\tw$ is treewidth, using an algorithm for $\LinE\CMSO$~\cite{ArnborgLS91,CourcelleM93,CourcelleMR00}.
Our method replaces the parameter $t_2$ by $|t_2 - \alpha_\phi(G,w')|$, where $\alpha_\phi(G,w')$ is the optimum value of $w'(Y)$ over all $Y$ such that $\phi(Y)$ holds.
Furthermore, our method allows us to handle a weighted version of the problem, which cannot be expressed by simply bounding the size of the solution set.
More formally, we obtain the following by \cref{thm:eval-argmso} and \cref{thm:eval-argmso-tw}.
\begin{definition}[Most vital nodes of $\delta$-approximate $\CMSO$ solutions]
    Let $\phi(X)$ be a $\CMSO$ formula with one free variable $X$.
    Given a graph $G$, weight functions $w, w':V(G) \to \Z$, and a parameter $\delta \geq 0$,
    \ProblemName{Most Vital Nodes of $\delta$-Approximate $\phi$-Solutions} asks for a minimum-weight vertex set $X$ with respect to $w$ such that
    \[
        \forall Y \subseteq V(G)\; \ab((\phi(Y) \land w'(Y) \leq \mathtt{opt} + \delta) \to X \cap Y \neq \varnothing),
    \]
    where $\mathtt{opt} = \min\{ w'(Y) \mid G \models \phi[Y] \}$.
\end{definition}
We define \ProblemName{Most Vital Edges of $\delta$-Approximate $\phi$-Solutions} analogously.

After determining $\delta'$ as described at the beginning of this section, the condition $\phi(Y) \land w'(Y) \leq \mathtt{opt} + \delta$ can be handled by an $\Am_{\delta}\CMSO$ formula, and we obtain the following corollaries.
\begin{corollary}\label{cor:interdiction-1}
    Let $\phi(X)$ be a $\CMSO_1$ formula with one free variable $X$.
    Let $G$ be a graph with clique-width $\cw$ and $\delta \geq 0$.
    Then, \ProblemName{Most Vital Nodes of $\delta$-Approximate $\phi$-Solutions} is fixed-parameter tractable parameterized by $\cw + \delta + |\phi|$.
\end{corollary}
\begin{corollary}\label{cor:interdiction-2}
    Let $\phi(X)$ be a $\CMSO_2$ formula with one free variable $X$.
    Let $G$ be a graph with treewidth $\tw$ and $\delta \geq 0$.
    Then, \ProblemName{Most Vital Nodes of $\delta$-Approximate $\phi$-Solutions} and \ProblemName{Most Vital Edges of $\delta$-Approximate $\phi$-Solutions} are fixed-parameter tractable parameterized by $\tw + \delta + |\phi|$.
\end{corollary}

\paragraph*{Examples}
The \ProblemName{Length-Bounded Cut} problem is a well-studied interdiction variant of the \ProblemName{Shortest Path} problem.
In \ProblemName{Length-Bounded Cut}, given a graph $G$, vertices $s, t \in V(G)$, and a positive integer $\lambda$, the task is to find a minimum-cardinality edge set $F\subseteq E(G)$ such that there is no $s$-$t$ path of length at most $\lambda$ in $G-F$.
The decision version of \ProblemName{Length-Bounded Cut} is NP-complete~\cite{BaierEHKKPSS10} even for $\lambda =4$,
W[1]-hard parameterized by the pathwidth plus the maximum degree~\cite{DvorakK18},
and can be solved in $\lambda^{O(\tw^2)}n$ time~\cite{DvorakK18}, where $\tw$ is the treewidth of the input graph.

If $s$ and $t$ are disconnected, or if their distance exceeds $\lambda$, the empty edge set is already a solution. Otherwise, let $d_{s,t}$ be the length of the shortest $s$-$t$ path in $G$ and set $\delta(\lambda)=\lambda-d_{s,t}\ge0$.
It is known~\cite{Courcelle97} that the property that an edge set $Y$ forms an $s$-$t$ path can be expressed by an $\MSO_2$ formula $\mathrm{Path}_{s,t}(Y)$.
Moreover, for an edge set $F$, an edge set $Y$ forms an $s$-$t$ path in $G-F$ if and only if $Y \cap F = \varnothing$ and $Y$ is an $s$-$t$ path in $G$.
Therefore, the desired edge set $X$ can be defined by an $\Am_{\delta(\lambda)}\MSO_2$-formula.
Thus, we obtain an FPT algorithm parameterized by $\delta(\lambda)$ plus treewidth by \cref{thm:eval-argmso-tw}.

\ProblemName{Minimum Spanning Tree Interdiction} is the interdiction variant of \ProblemName{Minimum Spanning Tree}.
Recall that a \emph{spanning tree} is an acyclic edge set that connects all vertices; this property is definable in $\MSO_2$.
\ProblemName{Minimum Spanning Tree Interdiction} is NP-hard in general~\cite{BallGV89},
and W[1]-hard parameterized by the weight $t_2$ of the spanning tree~\cite{GuoS14},
but \cref{cor:interdiction-2} applies.

Interdiction variants of \ProblemName{Matching}~\cite{Zenklusen10,GuoS14} and
\ProblemName{Independent Set}~\cite{BazganTT11} have also been studied.
These problems are NP-hard in general, but we can apply \cref{cor:interdiction-1} or \cref{cor:interdiction-2}.

\subsection{Forcing a unique optimum}\label{subsec:unique}
A \emph{forcing set} $F$ for perfect matchings is a subset of edges such that there is exactly one perfect matching $M$ with $F \subseteq M$.
The \emph{forcing number} for perfect matchings of a graph $G$ is the smallest cardinality of a forcing set, and it is NP-hard to compute the forcing number~\cite{AfshaniHM04}.
Note that the property of having a unique perfect matching is definable in $\MSO_2$, and thus by Courcelle's theorem~\cite{ArnborgLS91,CourcelleM93} computing the forcing number is fixed-parameter tractable parameterized by treewidth.

The notion of forcing sets has been generalized to other combinatorial optimization problems.
This has been investigated in various settings, such as puzzles~\cite{DemaineMSWA16}, graph coloring~\cite{HatamiM05}, SAT~\cite{HatamiM05,DemaineMSWA16}, minimum vertex cover~\cite{HoriyamaKOSS24,AnCCKLOS25}, and shortest paths~\cite{GimaKOS25}.
These problems often become $\Sigma^{\mathrm P}_2$-hard when the base problem is NP-hard.

For minimum vertex cover, this type of problem is called \ProblemName{Pre-assignment for Uniquification of Minimum Vertex Cover (PAU-VC)}, and it was shown that PAU-VC is fixed-parameter tractable parameterized by clique-width~\cite{AnCCKLOS25}.
Our results generalize this in the following sense.
\begin{definition}
    Let $\phi(X)$ be a graph property.
    \ProblemName{PAU-$\min$-$\phi$} is the problem that, given a graph $G$,
    asks for a minimum-cardinality set $S$ that is contained in exactly one minimum-cardinality set $U$ satisfying $\phi(U)$.
    We define \ProblemName{PAU-$\max$-$\phi$} analogously.
\end{definition}
The uniqueness property can be expressed by the $\MSO$ formula
\[
    \exists! X\,\phi(X) \equiv \exists X\,\forall Y\ab(\phi(Y) \leftrightarrow Y = X).
\]
Thus, by considering a formula $\psi(S) \equiv \exists! U \ab[U \in \Argmin(\phi(U)) \land S\subseteq U]$ and $S \in \Argmin(\psi)$, we obtain the following corollaries from \cref{thm:eval-argmso} or \cref{thm:eval-argmso-tw}.

\begin{corollary}
    Let $\phi(X)$ be a property definable by a $\CMSO_1$-formula with one free variable $X$.
    Then, \ProblemName{PAU-$\min$-$\phi$} and \ProblemName{PAU-$\max$-$\phi$} are fixed-parameter tractable parameterized by clique-width plus $|\phi|$.
\end{corollary}
\begin{corollary}
    Let $\phi(X)$ be a property definable by a $\CMSO_2$-formula with one free variable $X$.
    Then, \ProblemName{PAU-$\min$-$\phi$} and \ProblemName{PAU-$\max$-$\phi$} are fixed-parameter tractable parameterized by treewidth plus $|\phi|$.
\end{corollary}

\subsection{Diverse optimum solutions}
In diversity-maximization problems, the task is to find $r$ solutions that are diverse from one another.
The diversity measure is often defined by the Hamming distance between two solutions.
Let $A \Delta B = (A \setminus B) \cup (B \setminus A)$ denote the symmetric difference of sets $A$ and $B$; their Hamming distance is $|A\Delta B|$.
There are two commonly used diversity measures for a tuple $\vec S = (S_1, \dots, S_r)$ of $r$ $(\ge 2)$ solutions.
\begin{itemize}
    \item The sum of distances between all pairs of solutions: $\mathtt{SumHam}(\vec S) = \sum_{1 \le i < j \le r} |S_i \Delta S_j|$.
    \item The minimum distance between any pair of solutions: $\mathtt{MinHam}(\vec S) = \min_{1 \le i < j \le r} |S_i \Delta S_j|$.
\end{itemize}
Let $\phi(X)$ be a property of vertex sets or edge sets.
It is known that if $\phi$ is definable in $\CMSO_1$ (or $\CMSO_2$), then, given a graph $G$ and integers $r,k,d$, the problem of finding $r$ solutions $\vec S = (S_1, \dots, S_r)$ with $G\models \phi(S_i)$ and $|S_i| \leq k$ for all $i$ that have $\mathtt{SumHam}(\vec S) \geq d$ is fixed-parameter tractable parameterized by $r + k + \cw$ (or $r + k + \tw$)~\cite{Baste22,DrabikM26}, where $\cw$ and $\tw$ are the clique-width and treewidth of $G$, respectively.
When replacing $\mathtt{SumHam}$ with $\mathtt{MinHam}$, the problem is fixed-parameter tractable parameterized by $r + k + d+ \cw$ (or $r + k + d + \tw$)~\cite{Baste22,DrabikM26}.
These results are based on the notion of a \emph{dynamic-programming core}~\cite{Baste22,DrabikM26}, which permits a more precise running-time analysis.
Our method has the same polynomial factor in the running time but removes the dependence on the cardinality parameter $k$ when only nearly optimal solutions are considered.
Observe that $\mathtt{SumHam}(\vec S)$ is additive because it can be written as
\[
    \mathtt{SumHam}(\vec S) = \sum_{v \in V(G)} \sum_{1 \le i < j \le r} \indicator_{S_i \Delta S_j}(v) = \sum_{v \in V(G)} |\{\ell : v \in S_\ell\}|\cdot |\{\ell : v \notin S_\ell\}|.
\]
Let $f_w: S \mapsto \sum_{v \in S} w(v)$.  Then, the formula \[
    \psi(\vec S) \equiv \vec S \in \Argmax_{\mathtt{SumHam}}
    \ab(\bigwedge_{1 \le i \le r} S_i \in \Argmin_{f_w}(\phi))
\]
defines the set of $r$-tuples of solutions that maximize $\mathtt{SumHam}$ among all $r$-tuples of optimal solutions of $\phi$. Thus, we obtain the following corollary by \cref{thm:model-check-AM}.

\begin{corollary}\label{cor:sumham-amcmso}
    Let $G$ be a graph with clique-width $\cw$ and $w: V(G) \to \Z$ be a weight function.
    Let $\phi$ be a $\CMSO_1$ formula with one free variable $X$.
    Then, the problem of finding $r$ minimum-weight (or maximum-weight) solutions $\vec S = (S_1, \dots, S_r)$ with $G\models \phi[S_i]$ for all $i$ that maximize $\mathtt{SumHam}(\vec S)$ can be solved in $g(r,\cw,|\phi|)n^2$ time for some computable function $g$.
    In particular, maximizing $\mathtt{SumHam}$ among $r$ optimum solutions is fixed-parameter tractable parameterized by $r + \cw$ for the following problems:
    minimum weighted vertex cover, maximum weighted independent set, minimum weighted dominating set, minimum weighted feedback vertex set, and their connected variants.
\end{corollary}

For the measure $\mathtt{MinHam}$, the situation is more complicated since $\mathtt{MinHam}$ is not additive.
We use techniques similar to those in previous work~\cite{DrabikM26} to handle $\mathtt{MinHam}$.
Let $\mathbb B=(\{0,1\},\lor,\land)$ be the Boolean semiring and let $R=\mathbb B[(x_{ij})_{1\le i<j\le r}]$
be the Boolean polynomial semiring in $\binom r2$ variables $\ab<x_{ij}>_{1\le i < j\le r}$.
For an $r$-tuple $\vec S$, define
\[
    m(\vec S)=\prod_{1\le i<j\le r}
    x_{ij}^{|S_i\mathbin\Delta S_j|}.
\]
For a family $\mathscr A \subseteq \pset_r(V(G))$, define
    $h_m(\mathscr A)=\bigvee_{\vec S\in\mathscr A}m(\vec S)$.
We claim that $h_m$ is separably evaluable with respect to
$(R,\lor,\cdot,0)$.  Indeed, Boolean addition immediately gives
$h_m\left(\bigsqcup_{\ell}\mathscr A_\ell\right) =\bigvee_{\ell}h_m(\mathscr A_\ell)$
for pairwise disjoint families.  Moreover, if $\mathscr A$ and
$\mathscr B$ are separated, then their ground sets are disjoint and
$|(A_i\cup B_i)\mathbin\Delta(A_j\cup B_j)| =|A_i\mathbin\Delta A_j|+|B_i\mathbin\Delta B_j| $
for every $i<j$.
Hence, $h_m(\mathscr A\boxtimes\mathscr B) =h_m(\mathscr A)\cdot h_m(\mathscr B)$.
Therefore, $h_m$ is separably evaluable, and \cref{thm:evalargmso} gives an algorithm to compute $h_m$.
Then, we can compute the maximum $\mathtt{MinHam}(\vec S)$ in linear time of the size of $h_m$ by checking each monomial.

We now estimate the running time. Let $a=\binom r2$.
Since $|S_i \mathbin{\Delta} S_j|$ is at most $n = |V(G)|$, the degree in each variable of a polynomial $h_m(\mathscr A)$ is at most $n$.
Thus, addition takes $O(n^a)$ time, and multiplication takes $O(n^{2a})$ time by the naive method.
Since \cref{thm:evalargmso} uses $g(r,\cw,|\phi|)n$ polynomial operations,
the total running time is $g(r,\cw,|\phi|)\cdot O(n^{2a} n) \subseteq g(r,\cw,|\phi|)\cdot O(n^{r^2})$ for some computable function $g$.
Moreover, we can reduce the degree of each variable into $d$ by replacing the definition of $m(\vec S)$ with 
\[
m'(\vec S) = \prod_{1\le i<j\le r} x_{ij}^{\min(|S_i\mathbin\Delta S_j|, d)},
\]
and modify the multiplication operation to cap each resulting exponent at $d$.
The total running time of this approach is $g(r,\cw,|\phi|)\cdot O(d^{2a} n) \subseteq g(r,\cw,|\phi|)\cdot O(d^{r^2} n)$.

\begin{corollary}\label{cor:minham-amcmso}
    Let $\phi$ be a $\CMSO_1$ formula with one free variable $X$.
    Let $G$ be a graph with $n$ vertices with weight $w: V(G) \to \mathbb{Z}$.
    Let $d$ be a positive integer.
    Then, the problem of finding $r$ minimum (or maximum) weight solutions $\vec S = (S_1, \dots, S_r)$ with $G\models \phi[S_i]$ for all $i$ that maximize $\mathtt{MinHam}(\vec S)$ can be solved in $g(r,\cw,|\phi|)\cdot O(n^{r^2})$ time for some computable function $g$.
    For the decision version of the problem, where we seek $r$ solutions such that $\mathtt{MinHam}(\vec S) \ge d$, the problem can be solved in $g(r,\cw,|\phi|)\cdot O(d^{r^2}n + n^2)$ time.
\end{corollary}
Note that \cref{cor:sumham-amcmso,cor:minham-amcmso} can be generalized to the case of $\delta$-approximate solutions if $\delta$ is a parameter, where a $\delta$-approximate solution is a feasible solution such that the difference of its objective value from the optimal value is at most $\delta$.

\subsection{A note on applications of evaluation algorithms} \label{subsec:cardcmso}
This subsection explains how graph-polynomial evaluation algorithms apply to problems with a $\CMSO$ property $\phi$ and an additional cardinality constraint $\Gamma$. The results in this subsection reformulate known results, such as those in~\cite{Makowsky04,CourcelleD16}.

Some problems require a tuple of solutions $\vec S$ in which all sets have the same cardinality, that is, $|S_i| = |S_j|$ for all $i,j$.
For example, when the number of vertices is divisible by $k$, \ProblemName{Equitable $k$-Coloring} asks for a proper $k$-coloring of a graph such that all color classes have the same cardinality\footnote{In general, if the number of vertices is not divisible by $k$, \ProblemName{Equitable $k$-Coloring} requires that the $|(|S_i| - |S_j|)| \le 1$ for all $i,j$. This problem is also captured by \cref{thm:cardinality-constrained-cmso}.}.
It is well known that the property that $\vec S$ forms a proper $k$-coloring can be expressed by a $\CMSO$ formula $\phi(\vec S)$.
As described in \cref{subsec:eval-ammso}, we can construct the spectrum $s$ of the cardinalities of solutions $\vec S$ such that $G \models \phi(\vec S)$.
Then, we can check whether all color classes have equal size by testing whether, for some $q\in[0,n]$, the spectrum $s$ has a nonzero coefficient for the monomial $\prod_{i=1}^k x_i^q$.

In general, we can solve the following problems.
\begin{definition}
\ProblemName{Cardinality-Constrained $\Am\CMSO$ Problem} is defined by an $\Am\CMSO$ formula $\phi$ with $k$ free set variables $\vec X$ and a constraint $\Gamma \subseteq \N^k$.
We assume that $\vec a \in \Gamma$ can be tested in constant time for any $\vec a \in \N^k$.
Given a graph $G$, the task is to find a tuple of sets $\vec S$ such that $G \models \phi(\vec S)$ and $(|S_1|, \dots, |S_k|) \in \Gamma$.
For an $n$-vertex graph, the relevant restriction $\Gamma\cap[0,n]^k$ is given as a Boolean array indexed by $[0,n]^k$, with constant-time access to each entry.
\end{definition}

\cardinalityconstrainedcmso*

\begin{proof}
    Let $s$ be the spectrum of the cardinalities of solutions $\vec S$ such that $G \models \phi(\vec S)$.
    Then, we can check whether there is a tuple of solutions $\vec S$ such that $(|S_1|, \dots, |S_k|) \in \Gamma$ by checking whether the spectrum $s$ contains a tuple $(\gamma_1, \dots, \gamma_k) \in \Gamma$.
    Scanning the Boolean array and the corresponding spectrum entries takes $O((n+1)^k)$ time.
    A witnessing tuple can be recovered by standard dynamic-programming backtracking.

    Since two polynomials of degree at most $n$ in each of $k$ variables can be multiplied in $O(n^{2k})$ time by the naive method, and we may assume without loss of generality that the given $\cw$-expression of $G$ has $O(n)$ nodes, the total running time is $g(\cw,|\phi|)n^{2k+1}$ for some computable function $g$.
\end{proof}
The same argument gives a $\Am\CMSO_2$ variant on graphs of bounded treewidth.
For edge-set variables, the corresponding array indices range over $[0,|E(G)|]$; on graphs of bounded treewidth, this also has size $O(n)$ per coordinate.
This theorem can be applied to various problems, such as \ProblemName{Equitable $k$-Coloring}~\cite{BodlaenderF05}, \ProblemName{Equitable $k$-Partition}~\cite{BodlaenderF05}, \ProblemName{Partial Dominating Set}~\cite{AminiFS11}, and \ProblemName{Bisection}~\cite{HanakaKS21}.

\section{Hardness} \label{sec:hardness}

Let $\phi(X, Y)$ be a $\CMSO$ formula.
The formula $X \in \Argmin(\phi \mid Y)$ expresses that $X$ has minimum cardinality among the sets satisfying $\phi(X,Y)$ with $Y$ fixed.
The logic $1$-$\Am\CMSO$ extends $\Am\CMSO$ with such parameterized $\Argmin$ operators.

The following example is key to the proof of the NP-hardness of \ProblemName{$1$-$\Am\MSO_1$ Model Checking} on trees of depth 3.
\begin{lemma}\label{lem:definability}
	The following property is definable in $1$-$\Am\MSO_1$.
	\begin{itemize}
		\item $\mathtt{Same}(X)$: all connected components of the induced subgraph $G[X]$ have the same size.
	\end{itemize}
\end{lemma}
\begin{proof}
	The property $\mathtt{Same}(X)$ is equivalent to the condition that every connected component $Y$ of the induced subgraph $G[X]$ has minimum cardinality among the connected components of $G[X]$. Thus, we can define $\mathtt{Same}(X)$ as follows.
	\begin{align*}
		\mathtt{Same}(X) := \forall Y \ab(\mathtt{Conn}(X,Y) \to Y \in \Argmin(\mathtt{Conn}(X,Y) \mid X)).
	\end{align*}
	Here, $\mathtt{Conn}(X,Y)$ expresses that $Y$ is the vertex set of a connected component of the induced subgraph $G[X]$ and is definable in $\MSO_1$ (see, e.g.,~\cite{Libkin04}).
\end{proof}

\begin{theorem}\label{thm:nph}
	There exists a $1$-$\Am\MSO_1$ sentence $\phi$ such that determining whether a given graph $G$ satisfies $\phi$ is NP-hard even on trees of depth 3 with 4 colors.
	Moreover, $\phi$ has only one $\Argmin$ operator.
\end{theorem}
\begin{proof}
	\begin{figure}
		\centering
		\includegraphics{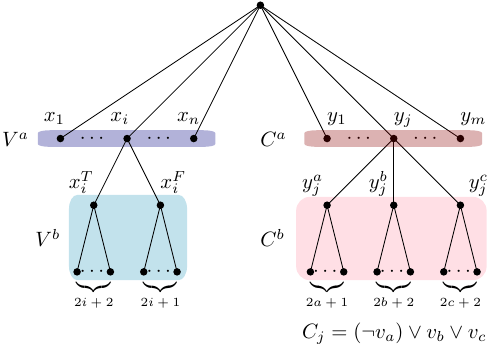}
		\caption{An illustration of the reduction in the proof of \cref{thm:nph}.}
		\label{fig:nphardness}
	\end{figure}
	We reduce from \ProblemName{3CNF-SAT}.
	Let $\psi$ be a 3CNF formula with clause set $\mathcal C = \{C_1, C_2, \ldots, C_m\}$ and variable set $V=\{v_1, v_2, \ldots, v_n\}$.
	We may assume that each clause contains three distinct variables.

	We construct a tree $T$ of depth 3 with 4 colors $\{V^a, V^b, C^a, C^b\}$ as follows.
	Create a root vertex $r$ and attach $n+m$ children $x_1, x_2, \ldots, x_n, y_1, y_2, \ldots, y_m$ to $r$.
	Here, $x_i$ has color $V^a$ and corresponds to the variable $v_i$, and $y_j$ has color $C^a$ and corresponds to the clause $C_j$.
	Then, for each $i \in [n]$, add two children $x_i^T$ and $x_i^F$, both with color $V^b$, to $x_i$.
	For each $i \in [n]$, add $2i+2$ children with color $V^b$ to $x_i^T$ and $2i+1$ children with color $V^b$ to $x_i^F$.
	Finally, for each $i \in [m]$ and for each literal $t \in C_i$, if $t = v_j$ for some $j \in [n]$, then add a child $y_i^j$ having color $C^b$ to $y_i$ and add $2j+2$ children having color $C^b$ to $y_i^j$; if $t = \lnot v_j$ for some $j \in [n]$, then add a child $y_i^j$ having color $C^b$ to $y_i$ and add $2j+1$ children having color $C^b$ to $y_i^j$.

    It is easy to see that the number of nodes in $T$ is polynomial in the number of variables and clauses of $\psi$, and hence the reduction can be computed in polynomial time.

	We use the following auxiliary formulas $\mathtt{Asg}(S)$ and $\mathtt{Sat}(S)$.
	\begin{itemize}
		\item $\mathtt{Asg}(S)$: $S \subseteq V^b$ consists of the vertices of exactly one of the two subtrees rooted at $x_i^T$ and $x_i^F$ for each $i \in [n]$.
		The intended meaning is that $S$ corresponds to a truth assignment to the variables.
		\item $\mathtt{Sat}(S)$: for each $i \in [m]$, there is a child $y_i^j$ of $y_i$ such that the induced subgraph $T[S]$ contains a connected component of the same size as the subtree rooted at $y_i^j$.
		The intended meaning is that the truth assignment corresponding to $S$ satisfies every clause $C_i$.
	\end{itemize}
	Before defining these formulas formally, we show that $T \models \exists S (\mathtt{Asg}(S) \land \mathtt{Sat}(S))$ if and only if the given 3CNF formula $\psi$ is satisfiable.

	Assume that $T \models \exists S (\mathtt{Asg}(S) \land \mathtt{Sat}(S))$.
	Then, there is a subset $S$ of $V^b$ such that $\mathtt{Asg}(S)$ and $\mathtt{Sat}(S)$ are true. Since $\mathtt{Asg}(S)$ is true, we can uniquely define a truth assignment $\sigma$ to the variables as follows: for each $i \in [n]$, if the subtree rooted at $x_i^T$ is contained in $S$, then $\sigma(v_i) = \mathtt{true}$; otherwise, $\sigma(v_i) = \mathtt{false}$.
	We show that $\sigma$ satisfies every clause $C_i$. Since $\mathtt{Sat}(S)$ is true, for each $i \in [m]$, there is a child $y_i^j$ of $y_i$ such that the induced subgraph $T[S]$ contains a connected component of the same size as the subtree rooted at $y_i^j$.
	By the construction of $T$, the subtree rooted at $y_i^j$ has the same size as the subtree rooted at $x_j^T$ if $C_i$ contains $v_j$ as a positive literal and has the same size as the subtree rooted at $x_j^F$ if $C_i$ contains $\lnot v_j$ as a negative literal.
	Thus, the truth assignment $\sigma$ satisfies $C_i$.
    Since this holds for every clause $C_i$, $\psi$ is satisfiable.

	Conversely, assume that the given formula $\psi$ is satisfiable and let $\sigma$ be a truth assignment to the variables that satisfies $\psi$.
	We define a subset $S$ of $V^b$ as follows: for each $i \in [n]$, if $\sigma(v_i) = \mathtt{true}$, then include all vertices of the subtree rooted at $x_i^T$ in $S$; otherwise, include all vertices of the subtree rooted at $x_i^F$ in $S$.
	Then, it is easy to see that $\mathtt{Asg}(S)$ is true.
    Let $i \in [m]$. Since $\sigma$ satisfies $C_i$, there is a literal $t \in C_i$ such that $\sigma(t) = \mathtt{true}$.
    By the construction of $S$, there is a connected component of $T[S]$ that has the same size as the subtree rooted at $y_i^j$, where $y_i^j$ is the child of $y_i$ corresponding to the literal $t$.
    Thus, $\mathtt{Sat}(S)$ is true, and hence $T \models \exists S (\mathtt{Asg}(S) \land \mathtt{Sat}(S))$.

	Finally, we define the formulas $\mathtt{Asg}(S)$ and $\mathtt{Sat}(S)$.
	The formula $\mathtt{Asg}(S)$ is defined as follows, where
	$\alpha \veebar \beta$ is the ``exclusive or'' operator, defined as $(\alpha \land \lnot \beta) \lor (\lnot \alpha \land \beta)$.
	\begin{align*}
		\mathtt{Asg}(S) :=  \forall v \forall x \forall y
		 \begin{bmatrix}
			&S \subseteq V^b \\
			\land& \begin{pmatrix}
				\ab(v \in V^a \land E(v,x) \land E(v,y) \land x \neq y \land x \in V^b \land y \in V^b) \\
				 \to (x \in S \veebar y \in S)
			\end{pmatrix} \\
			\land& \ab\big(x \in S \land E(x,y) \land y\in V^b \to y \in S)
		\end{bmatrix}.
	\end{align*}
	The meaning of $\mathtt{Asg}(S)$ is as follows: the first line means that $S$ is a subset of $V^b$; the second line means that, for each $i \in [n]$, exactly one of $x_i^T$ and $x_i^F$ belongs to $S$; the third line means that $S$ is closed under adjacency within $V^b$, so it contains every vertex of each selected subtree and no vertex of any unselected subtree.

	The formula $\mathtt{Sat}(S)$ is defined as follows.
	\begin{align*}
		\mathtt{Sat}(S) &:= \\ &
			\forall y \ab(y\in C^a \to \exists x\, \exists D\, \exists S'\, \exists W
		\begin{bmatrix}
			E(y, x) \land x \in C^b \\
			\land\; \forall z \ab(z \in D \leftrightarrow (z = x \lor (E(x,z) \land z \in C^b))) \\
			\land\; \mathtt{Conn}(S, S') \\
			\land\; \forall z\ab(z \in W \leftrightarrow (z \in S' \lor z \in D)) \\
			\land\; \mathtt{Same}(W)
		\end{bmatrix}).
	\end{align*}
	The meaning of $\mathtt{Sat}(S)$ is as follows.
	For each clause vertex $y \in C^a$, we find a literal vertex $x \in C^b$ adjacent to $y$, let $D$ be the vertex set of the subtree rooted at $x$ (i.e., $x$ together with all its $C^b$-children), and find a connected component $S'$ of $T[S]$.
	The set $W$ is defined as the disjoint union of $S'$ and $D$.
	Since $S \subseteq V^b$ and $D \subseteq C^b$ are vertex-disjoint with no edges between them in $T$, the induced subgraph $T[W]$ has exactly two connected components, $S'$ and $D$.
	Therefore, $\mathtt{Same}(W)$ expresses $|S'| = |D|$, i.e., $S'$ has the same size as the subtree rooted at $x$.
\end{proof}

\subsection{Hardness for the polynomial hierarchy}
We show that \ProblemName{$1$-$\Am\MSO_1$ Model Checking} is hard for every level of the polynomial hierarchy even on trees of bounded depth.

\begin{definition}
	Let $k \ge 1$ be an integer.
	\ProblemName{$\SigmaP_k$-3CNF-SAT} is defined as follows.
	The input is a Boolean formula $\phi$ of the form $\exists X_1 \forall X_2 \cdots Q_k X_k \psi(X_1,X_2, \ldots, X_k)$, where $\psi(X_1,X_2, \ldots, X_k)$ is a 3CNF formula whose variables are partitioned into sets $X_1, X_2, \ldots, X_k$, and $Q_k$ is $\exists$ if $k$ is odd and $\forall$ otherwise.
	The question is whether $\phi$ is true.

	\ProblemName{$\SigmaP_k$-3CNF-UNSAT} is defined as follows.
	The input is a Boolean formula $\phi$ of the form $\exists X_1 \forall X_2 \cdots Q_k X_k \lnot \psi(X_1,X_2, \ldots, X_k)$, where $\psi(X_1,X_2, \ldots, X_k)$ is a 3CNF formula whose variables are partitioned into sets $X_1, X_2, \ldots, X_k$, and $Q_k$ is $\exists$ if $k$ is odd and $\forall$ otherwise.
	The question is whether $\phi$ is true.
\end{definition}

\begin{theorem}[\cite{Stockmeyer76}] \label{thm:ph-sat}
	Let $k \ge 1$ be an integer.
	\begin{itemize}
			\item If $k$ is odd, then \ProblemName{$\SigmaP_k$-3CNF-SAT} is $\SigmaP_k$-hard.
			\item If $k$ is even, then \ProblemName{$\SigmaP_k$-3CNF-UNSAT} is $\SigmaP_k$-hard.
	\end{itemize}
\end{theorem}

\begin{corollary}\label{cor:phhardness}
	For any fixed integer $i \ge 1$,
	there exist $1$-$\Am\MSO_1$ sentences $\phi_s$ and $\phi_p$ such that
	\ProblemName{$1$-$\Am\MSO_1$ Model Checking} for $\phi_s$ is $\SigmaP_i$-hard even on trees of depth 3 with $O(i)$ colors,
	and \ProblemName{$1$-$\Am\MSO_1$ Model Checking} for $\phi_p$ is $\PiP_i$-hard even on trees of depth 3 with $O(i)$ colors.
\end{corollary}
\begin{proof}
    Suppose that $i$ is odd. Let $\psi$ be a 3CNF formula and consider an instance
		\[\exists X_1 \forall X_2 \cdots \exists X_i \psi(X_1,X_2, \ldots, X_i)\] of \ProblemName{$\SigmaP_i$-3CNF-SAT}.
    Let $T_{\psi}$ be a tree of depth 3 with $O(i)$ colors constructed from $\psi$ as in the proof of \cref{thm:nph}.
    For each $j \in [i]$, give each vertex in $V^b$ corresponding to a variable in $X_j$ the additional color $B_j$.
    We reuse the formula $\mathtt{Sat}(S)$ defined in the proof of \cref{thm:nph} and modify $\mathtt{Asg}(S)$ to $\mathtt{Asg}_j(S)$ for each $j \in [i]$ as described below.
    \begin{itemize}
	    \item $\mathtt{Asg}_j(S)$: $S \subseteq V^b \cap B_j$ consists of the vertices of exactly one of the two subtrees rooted at $x_\ell^T$ and $x_\ell^F$ for each $v_\ell \in X_j$.
    \end{itemize}
    Define\footnote{Here, $\mathtt{Sat}(\bigcup_{j=1}^i S_j)$ is the modified formula such that each $x \in U$ in $\mathtt{Sat}(U)$ is replaced with $\bigvee_{j=1}^i x \in S_j$.} 
     $\phi_i := \exists S_1 \;(\mathtt{Asg}_1(S_1) \land \forall S_2 \;(\mathtt{Asg}_2(S_2) \to (\cdots \exists S_i \;(\mathtt{Asg}_i(S_i) \land \mathtt{Sat}(\bigcup_{j=1}^i S_j)))))$.
    It is straightforward to see that $T_{\psi} \models \phi_i$ if and only if the given quantified Boolean formula is true. Thus, model checking for $\phi_i$ is $\SigmaP_i$-hard, and model checking for $\lnot \phi_i$ is $\PiP_i$-hard for odd $i$.
    The proof for even $i$ is similar by replacing the innermost existential block $\exists S_i \dots$ with a universal block $\forall S_i \;(\mathtt{Asg}_i(S_i) \to \lnot\mathtt{Sat}(\bigcup_{j=1}^i S_j))$.
\end{proof}

We can remove the colors in the above constructions by simulating them with attached leaves.
Assume $T$ has at most $c$ colors and these colors are indexed from $1$ to $c$.
If the color set of a vertex $v$ is $C \subseteq [c]$, we attach $\sum_{i\in C}2^i$ leaves to $v$.
If $C = \varnothing$, we attach one leaf to $v$.
Then, we can define the color of a vertex by the number of attached leaves using a first-order formula with $O(2^c)$ quantifiers.
We can check whether a vertex is a leaf using a first-order formula of constant size, and thus we can remove the colors in the above constructions by modifying the formulas and the tree $T$ accordingly.
The number of quantifiers and the size of the tree increase by a factor of $O(2^c)$, so the hardness conclusions of \cref{cor:phhardness} continue to hold after removing the colors. This modification increases the depth of the tree by only 1.
Hence, we obtain the following theorem.

\phhardargmso*

Lastly, we note that the above hardness results hold for any extension of $\MSO$ that can express the $\texttt{Same}(X)$ property in \cref{lem:definability}.

\section{Conclusion}\label{sec:conclusion}
We have shown that extending $\CMSO_1$ with $\Argmin$ and $\Argmax$ operators for additive objective functions preserves tractability on graphs of bounded clique-width, with a corresponding result for $\CMSO_2$ on graphs of bounded treewidth. This contrasts with extensions of $\MSO$ by cardinality comparisons, which are known to be intractable on graphs of bounded clique-width~\cite{Szeider11,DreierGH25}.
For many optimization problems involving optimal feasible solutions, our results remove the dependence on a cardinality parameter $k$.

Another interesting direction is to consider more general extensions that capture a meaningful class of problems.
Our extension $\Am\CMSO$ is equipped with an inverse image of the composition $\min \circ f$ (or $\max \circ f$) from the set of feasible solutions, where $f$ is an additive set function.
Is there a natural characterization of the class of functions $g$ or $g \circ f$ for which extending $\CMSO$ with a similar inverse image operator preserves tractability on graphs of bounded clique-width?
Moreover, it would be interesting to investigate whether such a characterization can define a meaningful class of problems.

From the viewpoint of combinatorial optimization, it is natural to ask whether width parameters are useful for solving other $\SigmaP_2$-type problems, such as Stackelberg games and minimization of maximum regret.

Another open question concerns the complexity of $\min$-$\max$-$\min$-type $\MSO_1$ problems.
The $\LinE\MSO_1$ framework yields an FPT algorithm for deciding $\exists X (|X| \le k \land \phi(X))$~\cite{CourcelleMR00}, while problems of the form $\exists X (|X| \le k_x \land \forall Y (|Y| \le k_y \to \phi(X,Y)))$ are in XP~\cite{DreierGH25}, parameterized by the clique-width of the input graph and the $\MSO_1$ formula $\phi$.
To the best of our knowledge, it remains open whether problems of the form $\exists X (|X| \le k_x \land \forall Y (|Y| \le k_y \to \exists Z (|Z| \le k_z \land \phi(X,Y,Z))))$ are in XP under the same parameterization.
In integer linear programming, feasibility of $\exists \vec x : A\vec x \le b$ and parametric feasibility of $\forall \vec y\in Q\; \exists \vec x : A\vec x + B\vec y \le b$ are decidable in polynomial time when the number of variables is fixed~\cite{Kannan90}. In contrast, deciding $\exists \vec x\in R\; \forall \vec y\in Q\; \exists \vec z : A\vec x + B\vec y + C\vec z \le b$ is NP-hard even when the number of variables is fixed~\cite{NguyenP22}. Here, all variables are integer vectors, and $Q$ and $R$ are rational polyhedra given by linear inequalities.
This raises the question of whether a similar hardness result holds for $\min$-$\max$-$\min$-type $\MSO_1$ problems, even when the clique-width of the input graph and $\phi$ are fixed.

\bibliography{references}

\end{document}